\documentclass[a4paper,UKenglish,cleveref, autoref, thm-restate]{lipics-v2021}

\pdfoutput=1 
\hideLIPIcs  

\nolinenumbers
\usepackage{amsmath}
\usepackage{amssymb}
\usepackage{xspace}
\usepackage{algorithm}
\usepackage[noend]{algpseudocode}
\usepackage{xparse} 

\usepackage{todonotes}
\usepackage{calc,ifthen}
\newcounter{hours}
\newcounter{minutes}
\newcommand{\Printtime}{\setcounter{hours}{\time/60}%
\setcounter{minutes}{\time-\value{hours}*60}%
\thehours:%
\ifthenelse{\value{minutes}<10}{0}{}\theminutes}
\usepackage[nodate]{datetime} 

\newcommand{\ALG}{\ensuremath{\operatorname{\textsc{Alg}}}\xspace}
\newcommand{\bALG}{\ensuremath{\operatorname{\textsc{Adaptive Threshold}}}\xspace}
\newcommand{\ubALG}{\ensuremath{\operatorname{\textsc{Adaptive Threshold for Untrusted Predictions}}}\xspace}
\newcommand{\OPT}{\ensuremath{\operatorname{\textsc{Opt}}}\xspace}
\newcommand{\ADB}{\ensuremath{\operatorname{\textsc{AT}}}\xspace}
\newcommand{\ADBML}{\ensuremath{\operatorname{\textsc{AT}}}\xspace}
\newcommand{\MLB}{\ensuremath{\operatorname{\textsc{ATup}}}\xspace}

\newcommand{\FLOOR}[1]{\left\lfloor#1\right\rfloor}
\newcommand{\SET}[1]{\left\{#1\right\}}

\newcommand{\NAT}{\ensuremath{\mathbb{N}}}
\newcommand{\eps}{\ensuremath{\varepsilon}}
\newcommand{\TF}{T}

\newcommand{\ASSIGN}{\ensuremath{\gets}\xspace}
\newcommand{\ITEMS}[1]{\ensuremath{n_{\!#1}}\xspace}
\newcommand{\SIZE}[1]{\ensuremath{\operatorname{\textrm{size}}(#1)}\xspace}
\newcommand{\LEVEL}{\ensuremath{\operatorname{\textrm{level}}}\xspace}

\newcommand{\guess}{\ensuremath{\hat{a}}\xspace}
\newcommand{\kfinal}{\ensuremath{k_t}\xspace}

\algnewcommand\algorithmicswitch{\textbf{switch}}
\algnewcommand\algorithmiccase{\textbf{case}}
\algdef{SE}[SWITCH]{Switch}{EndSwitch}[1]{\algorithmicswitch\ #1\ 
}{\algorithmicend\ \algorithmicswitch}%
\algdef{SE}[CASE]{Case}{EndCase}[1]{\algorithmiccase\ #1}{\algorithmicend\ \algorithmiccase}%
\algtext*{EndSwitch}%
\algtext*{EndCase}%
\makeatletter
\NewDocumentCommand{\LeftComment}{s m}{%
  \Statex \IfBooleanF{#1}{\hspace*{\ALG@thistlm}}\(\triangleright\) #2}
\makeatother

\title{Online Unit Profit Knapsack with Untrusted Predictions} 

\author{Joan Boyar}{University of Southern Denmark, Odense, Denmark \and \url{https://imada.sdu.dk/~joan/}}{joan@imada.sdu.dk}{https://orcid.org/0000-0002-0725-8341}{}

\author{Lene M. Favrholdt}{University of Southern Denmark, Odense, Denmark \and \url{https://imada.sdu.dk/~lenem/}}{lenem@imada.sdu.dk}{https://orcid.org/0000-0003-3054-2997}{}

\author{Kim S. Larsen}{University of Southern Denmark, Odense, Denmark \and \url{https://imada.sdu.dk/~kslarsen/}}{kslarsen@imada.sdu.dk}{https://orcid.org/0000-0003-0560-3794}{}

\authorrunning{J. Boyar, L.\,M. Favrholdt, and K.\,S. Larsen}

\Copyright{Joan Boyar, Lene M. Favrholdt, and Kim S. Larsen} 



\funding{Supported in part by the Independent Research Fund Denmark, Natural Sciences, grant DFF-0135-00018B.}

\begin{document}

\maketitle

\begin{abstract}
  A variant of the online knapsack problem is considered in the
  settings of trusted and untrusted predictions.  In Unit Profit
  Knapsack, the items have unit profit, and it is easy to find an
  optimal solution offline: Pack as many of the smallest items as
  possible into the knapsack.  For Online Unit Profit Knapsack, the
  competitive ratio is unbounded.
   In contrast,
   previous work on online algorithms with untrusted
   predictions generally
  studied problems where an online algorithm with a constant
  competitive ratio is known. 
  The prediction, possibly obtained from a machine learning source,
  that our algorithm uses is the average size of those smallest items
  that fit in the knapsack.  For the prediction error in this hard
  online problem, we use the ratio
  $r=\frac{a}{\guess}$ 
  where $a$ is the actual value for this average size and \guess is
  the prediction. The algorithm presented achieves a competitive ratio
  of $\frac{1}{2r}$ for $r\geq 1$ and $\frac{r}{2}$ for $r\leq
  1$. Using an adversary technique, we show that this is optimal in
  some sense, giving a trade-off in the competitive ratio attainable
  for different values of~$r$.  Note that the result for accurate
  advice, $r=1$, is only $\frac{1}{2}$, but we show that no algorithm
  knowing the value $a$ can achieve a competitive ratio better than
  $\frac{e-1}{e}\approx 0.6321$ and present an algorithm with a
  matching upper bound.
  We also show that
  this
  latter algorithm attains a competitive ratio of $r\frac{e-1}{e}$ for
  $r \leq 1$ and $\frac{e-r}{e}$ for
  $1 \leq r < e$, and no algorithm
  can be better for both $r<1$ and $1\leq r<e$.
\end{abstract}


\setcounter{page}{1}

\section{Introduction}
In this paper, we consider the Online Unit Profit Knapsack Problem:
The request sequence consists of $n$ item with sizes in $(0,1]$.  An online
algorithm receives them one at a time, with no knowledge of future
items, and makes an irrevocable decision for each, either accepting or
rejecting the item. It cannot accept any item if its size, plus the
sum of the sizes of the already accepted items, is greater than
$1$. The goal is to accept as many items as possible.  The obvious
greedy algorithm solves the offline Unit Profit Knapsack Problem,
since the set consisting of as many of the smallest items that fit in
the knapsack is an optimal solution.

Even for this special case of the Knapsack Problem, no competitive
online algorithms can exist.  Thus, we study the problem under the
assumption that (an approximation of) the average item size, $a$, in
an optimal solution is known to the algorithm.  We study the case,
where the exact value of $a$ is given to the algorithm as advice by an
oracle, as well as the case where $a$ is untrusted, e.g., estimated using
machine learning. For instance, the characteristics of the input may
be different depending on the time of day the input is produced, which
source produced the input, etc. This could be learned to some extent
and result in a prediction, which could be provided to the algorithm.

When considering machine-learned advice, the
concepts of consistency and robustness are often considered,
describing the balance between performing well on accurate advice and
not doing too poorly when the advice is completely wrong.  Our setting
is different from most work on online algorithms with machine-learned
advice, where there is generally a known online algorithm with a
constant competitive ratio for the problem without advice. For this
problem, if the advice is completely wrong, the algorithm cannot be
competitive, since the problem without advice does not allow for
competitive algorithms. Despite this hardness for the standard online
version of the problem, we obtain results with untrusted predictions
that are surprisingly consistent and robust.

\subsection{Previous Work}
\label{previouswork}
The Knapsack Problem is well studied and comes in many variants; see
Kellerer et al.~\cite{KPP04}.  Cygan et al.~\cite{CJS16} refer to the
online version we study, where all items give the same profit, as the
\emph{unit} case.  They mention that it is well-known that no online
algorithm for this version of the problem is competitive, i.e., has a
finite competitive ratio.  To verify this result, consider, for
instance, the family of input sequences $\sigma_j$ consisting of items
of sizes $\frac1i$, $i=1,2,3,\ldots,j$.

In the General Knapsack Problem, each item comes not only with a size,
but also with a profit, and the goal is to accept as much profit as
possible given that the total size must be at most~$1$.  The ratio of
the profit to the size is the \emph{importance} of an item.  (This is
sometimes called \emph{value}, but we want to avoid confusion with
other uses of that word.)

The Online Knapsack Problem was first studied by Marchetti-Spaccamela
and Vercellis~\cite{M-SV95}; they prove that the problem does not
allow for competitive online algorithms, even for Relaxed Knapsack
(fractions of items may be accepted), where all item sizes are $1$.
They concentrate on a stochastic version of the problem, where both
the profit and size coefficients are random variables.

The Online Unweighted (or Simple) Knapsack Problem with advice was
studied in~\cite{BKKR14}.
This is also called the proportional or
uniform case. In this version, the importance of each item is
equal to $1$.  They show that $1$ bit of advice is sufficient to
be $\frac12$-competitive, $\Omega(\log n)$ bits are necessary to be
better than $\frac12$-competitive, and $n-1$ advice bits are necessary
and sufficient to be optimal.
 (As mentioned later, they also considered the General
Knapsack Problem in the advice model.) The fundamental issues and  many of the early
results on oracle-based advice algorithms, primarily in the direction
of advice complexity, can be found in~\cite{BFKLM17}, though many
newer results for specific problems have been published since.

In~\cite{ZSHW20}, a knapsack problem is considered in a setting with
machine-learned advice, with results incomparable to ours. In their
setting, the General Knapsack Problem is considered, and results depend on upper
and lower bounds on the importance of the items.
The authors define
limited classes of algorithms, based on a parameter, leading to some
controlled degradation compared to an optimal competitive
ratio. Within the defined classes, focus is then on tuning compared
with historical data. Decisions to accept or reject an item are based
on a threshold function based on the item's importance.
Though the
definition of this function is ad hoc, in the sense that it is not
derived from some direct optimality criterion, it is well-motivated,
aiming to coincide with the behavior found in optimal algorithms for
the standard online algorithms setting.

Recently, in~\cite{IKQP21knapsack}, the General Knapsack Problem is
revisited, again with upper and lower bounds on the possible importance of items.
Machine-learned advice is given for each importance $v$, both an upper and
a lower bound for the sum of the sizes of the items with importance $v$.
The authors present an algorithm which has some similarities to ours.
In particular their budget function has a similar function to our
threshold function; both specify the maximum number of the low importance,
large items that need to be accepted to obtain the proven competitive
ratios. Their results can be extended to the case where the
predictions are off by a small amount, the lower bounds can be divided
by $1+\eps$, and the upper bounds can be multiplied by
$1+\eps$. This is in contrast to ours, where robustness results
are proven for arbitrarily large errors in the predictions, but only
$a$ is predicted. Since we have no bounds on the ratio of
the largest to smallest size, those values do not enter into our
results.  Their algorithm obtains what they prove to be the optimal
competitive ratio (for the given predictions), up to an additive
factor that goes to zero as the size of the largest item goes to zero;
this result has some of the flavor of our negative result.
The authors also consider two related problems.

The Bin Packing Problem is closely related to the Knapsack
Problem. This is
especially true for the dual variant where the number of bins is fixed
and the objective is to pack as many items as possible~\cite{BFLN01};
 the Unit Price Knapsack Problem is Dual Bin Packing
with one bin.
The standard Bin Packing Problem was considered with machine learning
in~\cite{ADJKR20}, considering a model of machine learning where, for
a given algorithm, \ALG, they consider a pair of values,
$(r_{\ALG},w_{\ALG})$, representing worst case ratios compared to the
optimal offline algorithm, \OPT. The value $r_{\ALG}$ gives the ratio
for the best (trusted) advice and $w_{\ALG}$ gives the ratio for the
worst possible (untrusted) advice. They use a parameter $\alpha$ in
their algorithm, and show that their algorithm achieves values
$(r, f(r))$ with $1.5<r \leq 1.75$ and $f(r)=\max\{33-18r,7/4\}$.

Bin Packing is also studied in~\cite{AKS21} in the standard setting
for online algorithms with machine learning, giving a trade-off
between consistency and robustness, with the performance degrading as
a function of the prediction error. They also have experimental
results. Since the problem is so difficult, they have restricted their
consideration to integer item sizes.

Much additional work has been done for other online problems, studying
variants with predictions (machine-learned advice, for instance),
initiated by the work of Lykouris and Vassilvitskii~\cite{LV18,LV21}
and Purohit et al.~\cite{PSK18} in 2018, with further work in the
directions of search-like
problems~\cite{A21,AKZ21,BCKP20,LHL19,LMHLSL21,MV17},
scheduling~\cite{AEMP22,AK21,BMRS20,IKQP21,KA18,LLMV20,LX21,M20},
rental problems~\cite{GP19,K19,WL20},
caching/paging~\cite{BCKPV22,IMMR20,JPS20,R20,W20}, and other
problems~\cite{AKS21,ACEPS20,AGKK20,BGGJ22,M21,RM21}, while some
papers attack multiple problems~\cite{ADJKR20,BMS20,LMRX21,WZ20}.  For
a survey, see~\cite{MV20}.

\subsection{Preliminaries}
We let $a$ denote the average size of items accepted by the offline,
optimal algorithm, \OPT, that accepts as many of the smallest items as
possible.  Moreover, we let $\guess$ denote the ``guessed'' or predicted
value of~$a$.  In the case of accurate advice (received from an
oracle), $\guess=a$.  If $\guess$ may not be accurate, possibly determined via machine learning, and
therefore not necessarily exactly~$a$, we define a ratio $r$ such that
$a = r \cdot \guess$.  This particular advice is considered as a value that
might be available or predictable, and the competitive ratios we
present are a function of~$r$.

We use the asymptotic competitive ratio throughout this paper.  Thus,
an algorithm \ALG, is $c$-competitive if there exists a constant $b$
such that for all request sequences $\sigma$,
$\ALG(\sigma) \geq c\OPT(\sigma)-b$, where $\ALG(\sigma)$ denotes
\ALG's profit on $\sigma$.  \ALG's competitive ratio is then
$\sup \{ c\mid \ALG \mbox{\rm ~is~}c\mbox{\rm -competitive}\}$.  Note
that this is a maximization problem and all competitive ratios are in
the interval~$[0,1]$.

We use the notation $\NAT=\SET{0, 1, 2, \ldots}$.  At any given time
during the processing of the input sequence, the {\em level} of the
knapsack denotes the total size of the items accepted.

\subsection{Our Results}
We consider both the case where the advice $\guess$ is known to be accurate, so $r=1$,
and the case where it might not be accurate. Different algorithms are
presented for these two cases, but they have a common form.

For our
algorithm \bALG (\ADB) where the advice is accurate and, thus, $\guess=a$, the
competitive ratio is $\frac{e-1}{e}$, and we prove a matching upper
bound that applies to any deterministic algorithm knowing~$a$. This
upper bound limits how well any algorithm using trusted predictions
can do; the competitive ratio cannot be better
than $\frac{e-1}{e} \approx 0.6321$ for $r=1$.

If \ADB is used for untrusted predictions, it obtains a competititve
ratio of $r\frac{e-1}{e}$ for $r \leq 1$, $\frac{e-r}{e}$ for $1 \leq
r \leq e$, and $0$ for $r \geq e$.
No algorithm can be better than this for both $r<1$ and $1 < r < e$.

For the results for
our algorithm, \ubALG (\MLB),
there
are two cases: for
$r\leq 1$ the competitive ratio is $\frac{r}{2}$, and for $r\geq 1$
the competitive ratio is $\frac{1}{2r}$.
Thus, for accurate
advice, the competitive ratio of \MLB is $\frac{1}{2}$, slightly less
good than for the other algorithm.
We show a negative result implying that an online algorithm
cannot both be $\frac{1}{2r}$-competitive for a range of large
$r$-values and better than $\frac12$-competitive for $r=1$.

Exact, oracle-based advice is not our focus point, though it is a
crucial step in our work towards an algorithm for untrusted
predictions. Thus, we do not emphasize the direction of advice
complexity, where the focus is on the number of bits of oracle advice
used to obtain given competitive ratios (or optimality),
but we include a brief discussion in Section~\ref{numberofbits}.
Instead, we
focus on advice that may be easy to obtain.  It seems believable
that the average size of requests in an optimal solution would be information easily
obtainable. The average size is probably a crucial component with
regards to the profit secured by a process and quite possibly crucial
with regards to supplying resources (knapsacks) over time. It is a
single number (or two numbers: number of items and total size) to
collect and store, as opposed to more detailed information about a
distribution. So little storage is required that one could keep
multiple copies if, for instance, the expected average changes during
the day.

Given the simple optimal algorithm for the offline version of unit
price knapsack, it seems obvious to consider another possibility for
advice, the maximum size, $s$, for items to accept. However, this is
insufficient, as there might be many items of that size, but the
optimal solution may contain very few of them. Thus, one also needs
further advice, including, for example, the fraction of the
knapsack filled by items of size~$s$. With these parameters given as advice,
there would be two possibilities for the error.
An extension of this idea is presented
in~\cite{BKKR14}, where the minimum importance is used, instead of the
maximum size, for the General Knapsack Problem, giving $k$-bit
approximations to the advice.

\section{The Adaptive Threshold Algorithm}

In Algorithm~\ref{algorithm-bounds}, we introduce an algorithm
template, which can be used to establish an oracle-based advice algorithm
as well as an algorithm for untrusted predictions.  The template omits the
definition of a threshold function, $\TF$, since it is different for the
two algorithms. In both algorithms, the threshold functions
 have the property that $\TF(i) > \TF(i+1)$ for
$i\geq 1$.  
We use the notation $\ITEMS{x}$ to denote the number of accepted items
strictly larger than $x$.

\begin{algorithm}
\begin{algorithmic}[1]
\State $\guess$ \ASSIGN predicted average size of \OPT's accepted items
\State \LEVEL \ASSIGN $0$
\For{each input item $x$}
   \State $i$ = $\max_{j\geq 0}\SET{\ITEMS{\TF(j+1)}=j}$ \label{alg:generici}
   \If{$\SIZE{x} \leq \TF(i+1)$ \textbf{and} $\LEVEL + \SIZE{x} \leq 1$}
      \State Accept $x$
      \State \LEVEL += $\SIZE{x}$
   \Else
      \State Reject $x$
   \EndIf
\EndFor
\end{algorithmic}
\caption{Algorithm \bALG.}
\label{algorithm-bounds}
\end{algorithm}

Intuitively, \bALG accepts items that fit as long as it has not
accepted too many items larger than the current item. The threshold
functions are used to determine how many larger items is too many;
no more than $i$ items of size larger than $\TF(i+1)$ are accepted. For smaller
item sizes, this number of larger items is larger, since we need to
accept more items if there are many small items.

Note that using $\max_{j\geq 0}\SET{\ITEMS{\TF(j+1)} \geq j}$ instead of
$\max_{j\geq 0}\SET{\ITEMS{\TF(j+1)}=j}$ in Line~\ref{alg:generici}
would result in the same algorithm.  Thus, $i$ is nondecreasing
through the processing of the input sequence,  and the value of the
threshold function, $\TF(i)$, is decreasing in $i$, so larger items cannot
be accepted after $i$ increases.

\section{Accurate Predictions}
\label{sec:advice}
In this section, we give an $\frac{e-1}{e}$-competitive algorithm
which receives $a$, the average size of the items in \OPT, as advice
and prove that it is optimal among algorithms that get only $a$ as
advice.

\subsection{Positive Result}

To define an advice-based algorithm, we define a threshold function;
see Algorithm~\ref{algorithm-advice}.
Throughout this section, we assume that $\guess=a$, but the algorithm
is also be used for untrusted predictions in Subsection~\ref{sec:semitrusted}.

\begin{algorithm}
\begin{algorithmic}[1]
\State Define $\displaystyle \TF(i)=\frac{\guess e}{\guess e(i-1)+1}$ for $i \geq 1$
\State Run \bALG, Algorithm~\ref{algorithm-bounds}
\end{algorithmic}
\caption{\bALG with advice, \ADB.}
\label{algorithm-advice}
\end{algorithm}

We first set out to prove that \ADB with $\guess = a$ has
competitive ratio at least $\frac{e-1}{e}\approx 0.6321$.
For that, we need two simple lemmas. The first involves an
obvious generalization of Harmonic numbers to non-integers.

Define $\sum_{i=x}^{y}f(k)$ for some function $f$ and real-valued $x$
and $y$ such that $y-x\in\NAT$ as $f(x)+f(x+1)+\cdots + f(y)$.  We
generalize the Harmonic numbers by defining
$H_k=\sum_{i=1+k-\FLOOR{k}}^{k}\frac{1}{k}$, for any real-valued
$k \geq 1$.

\begin{lemma}
\label{lemma-harmonic}
If $k \geq p \geq 1$ and $k-p\in\NAT$, then
$\ln k - \ln(p+1) \leq H_k - H_p \leq \ln k - \ln p$.
\end{lemma}
\begin{proof}
  Define $\Delta_k = H_k - \ln k$.  First, we argue that
  $\Delta_k > \Delta_{k+1}$.

Observe that
\[
\ln(k+1) - \ln k = \int_k^{k+1} \frac{1}{x}dx > \frac{1}{k+1},
\]
since $\frac{1}{k+1}$ is the smallest value we are integrating over.
So, $\frac{1}{k+1}-\ln(k+1) < -\ln(k)$.

Using this,
\[\Delta_{k+1}=H_{k+1}-\ln(k+1)=H_k+\frac{1}{k+1}-\ln(k+1)<H_k-\ln k=\Delta_k.\]
By the definition of $\Delta_k$,
\[
\begin{array}{rcl}
  H_k-H_p
  & = & \ln k + \Delta_k - (\ln p + \Delta_p) \\
  & \leq & \ln k - \ln p, \mbox{ since, by induction, $\Delta_k \leq \Delta_p$}\,.
\end{array}
\]

From the integral, it follows similarly that
$\ln(p+1)-\ln p < \frac{1}{p}$.  Thus,
\[\begin{array}{cl}
& \ln(p+1)-\ln p < \frac{1}{p} \\
\Updownarrow \\
& H_{p-1}-H_p=-\frac{1}{p} < \ln p - \ln(p+1) \\
\Updownarrow \\
& H_{p-1} - \ln p < H_p - \ln(p+1)\,.
\end{array}\]

Now, $H_k - \ln k > H_p - \ln(p+1)$ clearly holds for $p=k$, since
$\ln$ is increasing. So, by induction, using the above in the
induction step, it holds for smaller $p$ as well. Thus,
$\ln k - \ln(p+1) \leq H_k - H_p$ for $k\geq p$.
\end{proof}

The next lemma just establishes a simple analytical bound.
\begin{lemma}
\label{lemma-e-bound}
$\forall a > 0\colon e^{1-ae} \geq e - e^2a$.
\end{lemma}
\begin{proof}
  We prove that $\frac{e - e^{1-ae}}{a}$ is bounded from above by
  $e^2$.

  The derivative of the term is
  $\frac{ae^{2-ea}-(e-e^{1-ae})}{a^2} =
  \frac{e^{1-ae}(ea-e^{ea}+1)}{a^2}$.

  The terms $a^2$ and $e^{1-ae}$ are positive. Consider the remaining
  term, $ea-e^{ea}+1$. For $a=0$, this term is zero.  The derivative
  of $ea$ is $e$ and the derivative of $e^{ea}$ is $e^{ea+1}$.  For
  any $a > 0$, $e^{ea+1} > e$, so $ea-e^{ea}+1$ is negative.  Thus,
  for $a > 0$, the derivative of $\frac{e - e^{1-ae}}{a}$ is negative,
  and the term decreases with increasing~$a$.  Thus, the limit for $a$
  going towards zero is an upper bound.

  Using L'H\^{o}pital's rule,
  $ \lim_{a\rightarrow 0^{+}} \frac{e - e^{1-ae}}{a} =
  \lim_{a\rightarrow 0^{+}} \frac{e^{2-ae}}{1} = e^2.$
\end{proof}

With these two lemmas, we can now prove the theorem.

\begin{theorem}
  \label{thm:adb}
  For $\guess=a$, \ADB, as defined in Algorithm~\ref{algorithm-advice},
  is $\frac{e-1}{e}$-competitive.
\end{theorem}
\begin{proof}
  If \ADB never rejects an item, it performs optimally. So assume it
  rejects an item at some point in the request sequence~$\sigma$.
  Considering the conditional statement in the algorithm, if \ADB
  rejects an item, $x$, then either $\SIZE{x} > \TF(i+1)$ or
  $\LEVEL + \SIZE{x} > 1$.

\paragraph*{Case~1:}
  This is the case where, at some point, \ADB rejects an
  item, $x$, because $\LEVEL + \SIZE{x} > 1$. 

  The value of $\TF(k)$ from Algorithm~\ref{algorithm-bounds} is an
  upper bound on the size of the $k$th largest item accepted by the
  algorithm.  Thus, the $k$th largest accepted item has size at most
\[
  \TF(k) = \frac{ae}{ae(k-1)+1} = \frac{1}{k - 1 + \frac{1}{ae}}.
\]
Using the obvious definitions of sums over non-integer values,
as outlined above,
this gives
an upper bound on the total size of items accepted by \ADB of
\[
\LEVEL \leq
\sum_{k=1}^{\ADB(\sigma)}\frac{1}{k-1+\frac{1}{ae}} =
\sum_{k=\frac{1}{ae}}^{\ADB(\sigma)+\frac{1}{ae}-1}\frac{1}{k} =
H_{\ADB(\sigma)+\frac{1}{ae}-1} - H_{\frac{1}{ae} - 1}\,.
\]
Simple calculations (detailed in Lemma~\ref{lemma-harmonic}) give,
\[
H_{\ADB(\sigma)+\frac{1}{ae}-1} - H_{\frac{1}{ae} - 1} < 
\ln\left(\ADB(\sigma)+\frac{1}{ae}-1\right) -
\ln\left(\frac{1}{ae} - 1\right)
=  \ln\left(\frac{\ADB(\sigma)+\frac{1}{ae}-1}{\frac{1}{ae}-1}\right).
\]
By assumption, $\LEVEL + \SIZE{x} > 1$, and since
$\LEVEL
\leq
\ln\left(\frac{\ADB+\frac{1}{ae}-1}{\frac{1}{ae} - 1}\right)$,
we have
\begin{alignat*}{2}
& \ln\left(\frac{\ADB(\sigma)+\frac{1}{ae}-1}{\frac{1}{ae} - 1}\right)
  >
  1 - \SIZE{x} \\
\Updownarrow~~~ \\
& \frac{\ADB(\sigma)+\frac{1}{ae}-1}{\frac{1}{ae} - 1} >
  e^{1 - \SIZE{x}} \\
\Updownarrow~~~ \\
& \ADB(\sigma) > 
  \left( \frac{1}{ae} - 1 \right) e^{1 - \SIZE{x}}-\frac{1}{ae}+1\\
\Updownarrow~~~ \\
& \ADB(\sigma) > 
 \frac{e^{1 - \SIZE{x}}-1}{ae} + 1 - e^{1 - \SIZE{x}}\,.
\end{alignat*}

In the algorithm, $i$ is at least zero, so we cannot
accept items larger than $\TF(1)=ae$.
\begin{alignat*}{2}
  \ADB(\sigma) & > \frac{e^{1 - \SIZE{x}}-1}{ae} + 1 - e,
  &&  \mbox{ since $- e^{1 - \SIZE{x}} > -e$} \\
  & >  \frac{e^{1 - ae}-1}{ae} + 1 - e,
  &&  \mbox{ by the observation above} \\
  & \geq  \frac{e-e^2a-1}{ae} + 1 - e,
  &&  \mbox{ simple calculcations, detailed in Lemma~\ref{lemma-e-bound}} \\
  & =  \frac{e-1}{ae} - 2e + 1 \\
  & \geq  \frac{e-1}{e}\OPT(\sigma) - 2e + 1,
  &&  \mbox{ since $\OPT(\sigma) \leq \frac{1}{a}$}
\end{alignat*}
So, $\lim_{\OPT\to\infty}\frac{\ADB(\sigma)}{\OPT(\sigma)}\geq \frac{e-1}{e}$.

\paragraph*{Case~2:}
This is the case where \ADB never
rejects any item, $x$, when $\SIZE{x} \leq \TF(i+1)$.
Let $i_t$ denote the final value of $i$ as the algorithm terminates.
Suppose \OPT accepts $\ell$ items larger than $\TF(i_t+1)$ and $s$ items of
size at most $\TF(i_t+1)$.
Since \OPT accepts $\ell$ items larger than $\TF(i_t+1)$ and $\ell+s$
items in total, we have $a > \ell \cdot \TF(i_t+1) / (\ell+s)$, which is equivalent to
\begin{alignat}{2}
  \label{ineq:s}
  && s & > \left(\frac{\TF(i_t+1)}{a}-1\right)\ell
\end{alignat}

By the definition of $\TF$, we have that $\TF(i_t+1)=\frac{ae}{aei_t+1}$.
Solving for the $i_t$ on the right-hand side, we get
\begin{equation}
  \label{eq:it}
  i_t=\frac{1}{\TF(i_t+1)}-\frac{1}{ae} \,.
\end{equation}
Thus, \ADB has accepted at least
$i_t=\frac{1}{\TF(i_t+1)}-\frac{1}{ae}$ items of size greater than~$\TF(i_t+1)$.  Further, due to the
assumption in this second case, \ADB has accepted all of the $s$ items no larger
than~$\TF(i_t+1)$.  To see this, note that the $i$s of the algorithm can
only increase, so at no point has there been a size demand more
restrictive than $\TF(i_t+1)$.

We split in two subcases, depending on how $\TF(i_t+1)$ relates to \OPT's
average size, $a$.

\paragraph*{Subcase~2a: $\TF(i_t+1) > a$}
In this subcase, the lower bound on $s$ of Ineq.~(\ref{ineq:s}) is positive.
\begin{alignat*}{2}
  \frac{\ADB(\sigma)}{\OPT(\sigma)}
  & \geq 
  \frac{\left(\frac{1}{\TF(i_t+1)}-\frac{1}{ae}\right) + s}{\ell+s}\,, &&
  \text{ by Eq.~(\ref{eq:it})} \\
  & > 
  \frac{\left(\frac{1}{\TF(i_t+1)}-\frac{1}{ae}\right) +
    \left(\frac{\TF(i_t+1)}{a}-1\right)\ell}{\ell+\left(\frac{\TF(i_t+1)}{a}-1\right)\ell}\,, &&
  \text{ by Ineq.~(\ref{ineq:s})}\\
  & = 
  \frac{\left(\frac{1}{\TF(i_t+1)}-\frac{1}{ae}\right) + \left(\frac{\TF(i_t+1)}{a}-1\right)\ell}{\frac{\TF(i_t+1)}{a}\ell}\,.
  \end{alignat*}
The second inequality follows since the ratio is smaller than one and $s$ is
replaced by a smaller, positive term in the numerator as well as the
denominator.
  
We prove that this is bounded from below by $\frac{e-1}{e}$:
\begin{align*}
& \frac{\left(\frac{1}{\TF(i_t+1)}-\frac{1}{ae}\right) + \left(\frac{\TF(i_t+1)}{a}-1\right)\ell}{\frac{\TF(i_t+1)}{a}\ell}
  \geq \frac{e-1}{e} \\
\Updownarrow~~~ \\
& \frac{e}{\TF(i_t+1)} - \frac{1}{a} + e\left(\frac{\TF(i_t+1)}{a}-1\right)\ell
  \geq e\frac{\TF(i_t+1)}{a}\ell - \frac{\TF(i_t+1)}{a}\ell \\
\Updownarrow~~~ \\
& \frac{e}{\TF(i_t+1)} - \frac{1}{a}
  \geq  \left(e - \frac{\TF(i_t+1)}{a}\right)\ell \\
\Updownarrow~~~ \\
& \frac{ea-\TF(i_t+1)}{a\TF(i_t+1)}
  \geq \frac{ea-\TF(i_t+1)}{a}\ell \\
\Updownarrow~~~ \\
& \frac{1}{\TF(i_t+1)} \geq \ell
\end{align*}
For the last biimplication, we must argue that $ea-\TF(i_t+1) \geq 0$,
but this holds since $\TF(1)=ea$ and $\TF$ is decreasing.
Finally, the last statement, $\frac{1}{\TF(i_t+1)} \geq \ell$ holds
regardless of the relationship between $\TF(i_t+1)$ and $a$, since the
knapsack obviously cannot hold more than $\frac{1}{\TF(i_t+1)}$ items of
size greater than~$\TF(i_t+1)$.

\paragraph*{Subcase~2b: $\TF(i_t+1) \leq a$}
\begin{alignat*}{2}
  \frac{\ADB(\sigma)}{\OPT(\sigma)}
  & \geq
  \frac{\left(\frac{1}{\TF(i_t+1)}-\frac{1}{ae}\right) + s}{\ell+s}\,,
  &&\text{ by Eq.~(\ref{eq:it}) }\\
  & \geq
  \frac{\left(\frac{1}{\TF(i_t+1)}-\frac{1}{ae}\right)}{\ell}\,, && \text{
    since } s \geq 0 \text{ and } \frac{\ADB(\sigma)}{\OPT(\sigma)}
  \leq 1 \\
  & >
  \frac{\left(\frac{1}{\TF(i_t+1)}-\frac{1}{ae}\right)}{\frac{1}{\TF(i_t+1)}}\,,
  &&\text{ since, as above, $\ell \leq \frac{1}{\TF(i_t+1)}$ } \\
  & =
  1 - \frac{\TF(i_t+1)}{ae} \\
  & \geq
  1 - \frac{a}{ae}\,,&& \text{ by the subcase we are in} \\
  &  =
  \frac{e-1}{e}\,.
\end{alignat*}
This concludes the second case, and, thus, the proof.
\end{proof}

\subsection{Negative Result}
Now, we show that \ADB is optimal among
online algorithms knowing $a$ and nothing else.
\begin{theorem}
\label{thm:advneg}
  Any algorithm getting only $a$ as advice has a competitive ratio of
  at most~$\frac{e-1}{e}$.
\end{theorem}

\begin{proof}
  Let \ALG denote the online algorithm with advice, and let $\sigma$
  be the adversarial sequence defined by
  Algorithm~\ref{algorithm-adversarial}, which explains how the
  adversary defines its sequence based on \ALG's actions.

\begin{algorithm}
  \begin{algorithmic}[1]
    \LeftComment{Assume $a < \frac{1}{2e}$ and $\frac1a \in \mathbb{N}$}
    \State $\eps$ \ASSIGN $\frac{a^2}{10}$ 
    \State $k$ \ASSIGN $\FLOOR{\frac{1}{ae}}$ \linespread{1.2}\selectfont
    \While{\ALG's $\LEVEL \leq 1-\frac{1}{k}-k\eps$} \linespread{1}\selectfont
       \For{$k$ times}
          \State Give an item of size $\frac{1}{k}-\eps$
          \If{\ALG accepts}
             \State $k$++
             \State \textbf{continue} (* the while-loop *)
          \EndIf
       \EndFor
       \LeftComment{\ALG did not accept any of the $k$ items of this round.}
       \State Give $\frac{1}{a}-k$ items of size $\frac{ka\eps}{1-ka}$\label{caseone}
       \State \textbf{terminate} \Comment{Case 1}
    \EndWhile
    \State Give $\frac{1}{a}$ items of size $a$ \Comment{Case 2}
\end{algorithmic}
\caption{Adversarial sequence establishing optimality with advice.}
\label{algorithm-adversarial}
\end{algorithm}

Let $\kfinal$ be the value of $k$ at the beginning of the last iteration of
the while-loop.
We perform a case analysis based on how the generation of the
adversarial sequence terminates.

\paragraph*{Case~1:}
\OPT accepts the $\kfinal$ items of size~$\frac{1}{\kfinal}-\eps$ in the last
iteration of the while-loop and the $\frac{1}{a}-\kfinal$ items of
size~$\frac{\kfinal a\eps}{1-\kfinal a}$ for a total of $\frac{1}{a}$ items of total
size
\[\kfinal\left(\frac{1}{\kfinal}-\eps\right)+\left(\frac{1}{a}-\kfinal\right)\frac{\kfinal a\eps}{1-\kfinal a} =
  1-\kfinal\eps+(1-\kfinal a)\frac{\kfinal\eps}{1-\kfinal a} = 1\,.\]
Note that the average size of the items accepted by \OPT is $a$,
consistent with the advice.

\ALG accepts one item in each iteration of the while-loop, except the
last iteration, and at most $\frac{1}{a}-\kfinal$ items after that, so no
more than
\[\kfinal-\FLOOR{\frac{1}{ae}} + \frac{1}{a}-\kfinal <
  \frac{1}{a}-\frac{1}{ae}+1 =
  \frac{e-1}{e}\cdot\frac{1}{a}+1\,.\]
\[\mbox{Thus, }
  \ALG(\sigma) \leq \frac{e-1}{e}\cdot\frac{1}{a}+1 = \frac{e-1}{e} \OPT(\sigma)+1\,.\]

\paragraph*{Case~2:}
\OPT accepts the $\frac{1}{a}$ items of size~$a$.

For the analysis of \ALG, we start by establishing an upper bound
on~$\kfinal$.  The following inequality holds since \ALG accepts one item
per round, and \ALG's level just before the last round is at
most~$1-\frac{1}{\kfinal}-\kfinal\eps$ before the last item of size
$\frac{1}{\kfinal}-\eps$ is accepted.
\begin{align*}
  & \sum_{k=\FLOOR{\frac{1}{ae}}}^{\kfinal}\left(\frac{1}{k}-\eps\right) \leq 1-(\kfinal+1)\eps \\
\Downarrow~~~ \\
  & H_{\kfinal} - H_{\FLOOR{\frac{1}{ae}}-1} -\kfinal\eps < 1 - \kfinal\eps \\
\Updownarrow~~~ \\
  & H_{\kfinal} - H_{\FLOOR{\frac{1}{ae}}-1} < 1 \\
\Downarrow~~~ \\
  & \ln(\kfinal) - \ln\left(\FLOOR{\frac{1}{ae}}\right) < 1, \mbox{ simple calculcations, detailed in Lemma~\ref{lemma-harmonic}} \\
\Updownarrow~~~ \\
  & \kfinal < e\FLOOR{\frac{1}{ae}}
\end{align*}

In the case we are treating, \ALG leaves the while-loop because its
level is more than $1-\frac{1}{\kfinal+1}-(\kfinal+1)\eps$.  Now, we give a bound
on the amount of space available at that point.  For the first
inequality, note that by the initialization of $k$ in the algorithm,
$\kfinal \geq \FLOOR{\frac{1}{ae}}$.

\begin{align*}
\frac{1}{\kfinal+1}+(\kfinal+1)\eps
& <  \frac{1}{\FLOOR{\frac{1}{ae}}+1} + \left(e\FLOOR{\frac{1}{ae}}+1\right)\eps 
 <  ae + \left(\frac{1}{a} + 1\right)\frac{a^2}{10} \\
& <  \left(e + \left(\frac{1+a}{10}\right)\right)a 
 <  3a
\end{align*}

Thus, after the while-loop, \ALG can accept at most two of the items
of size~$a$.  Clearly, the number of rounds in the while-loop is
$\kfinal-\FLOOR{\frac{1}{ae}}+1$.  Using $\kfinal < e\FLOOR{\frac{1}{ae}}$,
we can now bound \ALG's profit:
\[\ALG(\sigma)\leq \kfinal-\FLOOR{\frac{1}{ae}}+1+2<(e-1)\FLOOR{\frac{1}{ae}}+3\leq \frac{e-1}{e}\frac{1}{a}+3
=\frac{e-1}{e}\OPT(\sigma)+3\]

This establishes the bound on the competitive ratio of $\frac{e-1}{e}$.

Finally, to ensure that our proof is valid, we must argue that the
number of rounds we count in the algorithm and the sizes of items we
give are non-negative. For the remainder of this proof, we go through the terms in the algorithm,
thereby establishing this.

The largest value of $k$ in the algorithm is $\kfinal$, and we have
established that $\kfinal<e\FLOOR{\frac{1}{ae}}<\frac{1}{a}$.  Additionally,
from the start value of $k$, we know that
$\FLOOR{\frac{1}{ae}}\leq k$. Using these facts, together with the
assumption from the algorithm that $a<\frac{1}{2e}$, we get the
following bounds on the various terms.
\[
  1-\frac{1}{k}-k\eps
  > 1-\frac{1}{\FLOOR{\frac{1}{ae}}}-\frac{1}{a}\frac{a^2}{10}
  > 1-\frac{1}{\FLOOR{\frac{1}{1/2}}}-\frac{1}{20e}
  > 0
\]
Further, $\frac{1}{k}-\eps\geq\frac{1}{\kfinal}-\eps>\frac{1}{\frac{1}{a}}-\frac{a^2}{10}>0$ and $\frac{1}{a}-k\geq \frac{1}{a}-\kfinal>\frac{1}{a}-\frac{1}{a}=0$.

For the last relevant value, $1-ka\geq 1-\kfinal a>1-\frac{1}{a}a=0$ and from Case~1, we know
that the $\frac{1}{a}-\kfinal$ items given in Line~\ref{caseone} of the
algorithm sum up to at most one.
\end{proof}

\section{Untrusted Predictions}

For the case where the predictions may be inaccurate, the algorithm \ADB
can be used with \guess possibly not being $a$ as long as $r < e$, see
Subsection~\ref{sec:semitrusted}.
In Subsection~\ref{sec:MLalg}, we give an adaptive threshold
algorithm, \MLB, that works for all $r$.

For $r < \frac12(e+\sqrt{e^2-2e}) \approx 2.06$,
\ADB has a better competitive ratio than \MLB.
Thus, if an upper bound on $r$ of approximately $2$ (or lower) is known,
\ADB may be preferred, and if a
guarantee for any $r$ is needed, \MLB should
be used.

\subsection{Semi-Trusted Predictions}
\label{sec:semitrusted}

In this section, we consider the algorithm \ADB with a semi-trusted
(being guaranteed that $r<e$) prediction, \guess, instead of $a$.

\subsubsection{Positive Result}
\label{sec:advicealg}
In this section, we consider the algorithm \ADB with \guess instead of $a$.
Note that the lower bound of the theorem below is positive only when
$r < e$. For $r \geq e$, the algorithm may not accept any items, and,
hence, its competitive ratio is $0$.

\begin{theorem}
  \label{thm:adbp}
  For untrusted advice, \ADBML 
  has a competitive ratio of at least
  \begin{align*}
    c_{\ADB}(r) \geq
    \begin{cases}
      \displaystyle \frac{e-1}{e} \cdot r, & \text{if }r \leq 1\\
      \displaystyle \frac{e-r}{e}, & \text{if } r \geq 1
    \end{cases}
    \end{align*}
\end{theorem}

\begin{proof}
The proof is analogous to the proof of Theorem~\ref{thm:adb}.

In Case~1, replacing $a$ by \guess, since the algorithm bases its actions on \guess
instead of $a$, and setting $\OPT =
\frac{1}{r\guess}$, results in a ratio of
$$\frac{\ADB(\sigma)}{\OPT(\sigma)} \geq \frac{e-1}{e} \cdot r$$ 
instead of $\frac{e-1}{e}$.

In Case~2, the lower bound on $s$ given in Ineq.~(\ref{ineq:s}) depends on
the actual average size, $a$, whereas the value of $i_t$ given
in Eq.~(\ref{eq:it}) depends on \guess, since the algorithm uses \guess.
The subcase distinction is still based on $a$.

In Subcase~2a, we obtain
\begin{align*}
  \frac{\ADB(\sigma)}{\OPT(\sigma)}
  & \geq 
  \frac{\left(\frac{1}{\TF(i_t+1)}-\frac{1}{\guess e}\right) + s}{\ell+s}
  \geq 
  \frac{\left(\frac{1}{\TF(i_t+1)}-\frac{1}{\guess e}\right) + \left(\frac{\TF(i_t+1)}{a}-1\right)\ell}{\frac{\TF(i_t+1)}{a}\ell}\,.
\end{align*}
Going through the same calculations as in the proof of
Theorem~\ref{thm:adb}, we get that 
\begin{alignat*}{2}
&&& \frac{\left(\frac{1}{\TF(i_t+1)}-\frac{1}{\guess e}\right) +
    \left(\frac{\TF(i_t+1)}{a}-1\right)\ell}{\frac{\TF(i_t+1)}{a}\ell}  \geq \frac{e-r}{e} \\
&\Updownarrow~~~ && \\
&&&  \frac{e}{\TF(i_t+1)} - \frac{1}{\guess}
  \geq  \left(e - r \cdot \frac{\TF(i_t+1)}{r \guess}\right)\ell \\
&\Updownarrow~~~ && \\
&&&  \frac{1}{\TF(i_t+1)} \geq \ell \,.
\end{alignat*}

In subcase~2b, we obtain
\begin{align*}
  \frac{\ADB(\sigma)}{\OPT(\sigma)}
  & \geq
  \frac{\left(\frac{1}{\TF(i_t+1)}-\frac{1}{\guess e}\right) + s}{\ell+s}
    > 1 - \frac{\TF(i_t+1)}{\guess e} 
   \geq  1 - \frac{a}{\guess e}
   = 1 - \frac{r \guess}{\guess e}
   = \frac{e-r}{e}\,.
\end{align*}

Thus, we obtain a lower bound of $\frac{e-1}{e} \cdot r$ in Case~1 and
a lower bound of $\frac{e-r}{e}$ in Case~2.
For $r \leq 1$, $\frac{e-1}{e} \cdot r \leq \frac{e-r}{e}$, and for $r \geq 1$, $\frac{e-1}{e} \cdot r \geq \frac{e-r}{e}$.
\end{proof}

\subsubsection{Negative Result}

The following result shows that, for $r<e$, no algorithm can be better
than \ADB for both $r < 1$ and
$r > 1$.

\begin{restatable}{theorem}{generaluppersmallr}
  \label{thm:generaluppersmallr}
  If an algorithm is $\frac{e-r}{e}$-competitive for all $1 \leq r<e$,
  it cannot be better than $r \cdot \frac{e-1}{e}$-competitive for any
  $r \leq 1$.
  If an algorithm is better than $r \cdot \frac{e-1}{e}$-competitive for some
  $r \leq 1$, it cannot be $\frac{e-r}{e}$-competitive for all $1 \leq r<e$.
\end{restatable}
\begin{proof}
  Consider an algorithm, \ALG.
  
  Assume that \ALG is
$\frac{e-r}{e}$-competitive for all $1 \leq r<e$. Then there exists a
  constant, $b$, such that $\ALG(\sigma) \geq 
\frac{e-r}{e} \OPT(\sigma)-b$, for any sequence $\sigma$ and any $1
\leq r < e$.
This constant $b$ is given as a parameter to
Algorithm~\ref{algorithm-adversarial-semi}, constructing an adversarial sequence, $\sigma$.
\begin{algorithm}
  \begin{algorithmic}[1]
    \LeftComment{Assume $\guess < \frac{1}{2e+b}$ and $\frac{1}{r_2 \guess} \in \mathbb{N}$}
    \State $k$ \ASSIGN $\FLOOR{\frac{1}{\guess e}}-1$ \linespread{1.2}\selectfont
    \While{\ALG's $\LEVEL \leq 1-\frac{1}{k+1}-(b+1)\guess e$ and $\frac{1}{k+1} \geq \guess$} \linespread{1}\selectfont
       \State $k$++
       \For{$k$ times}
          \State Give an item of size $\frac{1}{k}$
          \If{\ALG accepts}
             \State \textbf{continue} (* the while-loop *)
          \EndIf
       \EndFor
       \If{\ALG has accepted fewer than $k-\FLOOR{\frac{1}{\guess e}}-b$ items}
          \State {\bf terminate}\label{line:toofew}
       \EndIf
    \EndWhile
    \State Give $\frac{1}{r_2 \guess}$ items of size $r_2 \guess$\label{line:last}
\end{algorithmic}
\caption{Adversarial sequence for $r<e$. The 
adversarial algorithm takes two parameters, $r_2<1$ and $b \geq 0$.}
\label{algorithm-adversarial-semi}
\end{algorithm}

Let \kfinal be the value of $k$ at the end of the last iteration
of the while-loop. 

If the adversarial algorithm terminates in Line~\ref{line:toofew},
then \ALG has accepted at most $\kfinal-\FLOOR{\frac{1}{\guess e}}-b-1$
items. For termination in Line~\ref{line:toofew}, \ALG has not accepted
any of the $\kfinal$ items in the for-loop immediately preceding this,
so $\kfinal$ items of size $\frac{1}{\kfinal}$ were given.
In this case, \OPT accepts exactly these \kfinal items from the last
iteration of the while-loop, and $a=\frac{1}{\kfinal}$.
Let $r_1 = a/\guess$. Since $\frac{1}{\kfinal}  \geq \guess$, $r_1 \geq 1$.
Then,
\begin{align*}
  \ALG(\sigma) & \leq \kfinal-\FLOOR{\frac{1}{\guess e}}-b-1
   < \OPT - \frac{1}{\guess e} - b 
   = \OPT - \frac{r_1}{ae} - b 
   = \OPT - \frac{r_1}{e}\OPT - b \\
  & = \frac{e-r_1}{e} \OPT - b,
\end{align*}
contradicting that $\ALG(\sigma) \geq \frac{e-r}{e} \OPT(\sigma)-b$.
Thus, 
the adversarial algorithm cannot terminate in Line~\ref{line:toofew}.

Since the adversarial algorithm does not terminate in
Line~\ref{line:toofew},
it must accept its first item no later than in the $(b+2)$nd iteration
of the while-loop, and
the $i$th item accepted by \ALG has size at least
$\frac{1}{\FLOOR{\frac{1}{\guess e}}+b+i} \geq \frac{1}{\frac{1}{\guess e}+b+i}$.
Thus, the total size, $S_t$, of the items accepted by \ALG in the while-loop
is
\begin{alignat*}{2}
  S_t & \geq \sum_{k=\frac{1}{\guess e}+b+1}^{\kfinal} \frac{1}{k}\\
  & =  \sum_{k=1}^{\kfinal} \frac{1}{k}-
       \sum_{k=1}^{\frac{1}{\guess e}-1} \frac{1}{k}-
       \sum_{k=\frac{1}{\guess e}}^{\frac{1}{\guess e}+b} \frac{1}{k}  \\
  & > H_{\kfinal} - H_{\frac{1}{\guess e}-1} - (b+1)\guess e\\
  & \geq \ln(\kfinal) - \ln\left(\frac{1}{\guess e} \right) - (b+1)\guess e,
  && \text{ by Lemma~\ref{lemma-harmonic}.}
\end{alignat*}
Thus, we have
\begin{align}
  \label{ineq:St}
  S_t > \ln(\kfinal) - \ln\left(\frac{1}{\guess e} \right) - (b+1)\guess e\,.
\end{align}
By the first condition of the while-loop, and since \ALG accepts at
most one item per iteration, $S_t \leq 1-(b+1)\guess e$.
By Ineq.~(\ref{ineq:St}), this means that $\ln(\kfinal) -
\ln\left(\frac{1}{\guess e} \right) - (b+1)\guess e <
1-\frac{1}{\kfinal}-(b+1)\guess e$, and we get
\begin{align}
& \ln(\kfinal) -\ln\left(\frac{1}{\guess e} \right) - (b+1)\guess e <
1-(b+1)\guess e\nonumber\\
\Updownarrow~~~\nonumber\\
& \ln(\kfinal) -\ln\left(\frac{1}{\guess e} \right) < 1\nonumber\\
\Updownarrow~~~\nonumber\\
& \ln \left( \frac{\kfinal}{\frac{1}{\guess e}} \right) < 1\nonumber\\
\Updownarrow~~~\nonumber\\
& \frac{\kfinal}{\frac{1}{\guess e}} < e\nonumber\\
\Updownarrow~~~\nonumber\\
& \kfinal < \frac{1}{\guess} \label{ineq:ktupper} \,.
\end{align}

Furthermore, by the conditions of the while-loop, we have that 
$S_t > 1-\frac{1}{\kfinal+1} - (b+1)\guess e$ or
$\frac{1}{\kfinal+1} < \guess$.

If $S_t > 1-\frac{1}{\kfinal+1} - (b+1)\guess e$, then,
using that $k_t\geq \FLOOR{\frac{1}{\guess e}}$,
$$S_t >  1-\frac{1}{ \FLOOR{\frac{1}{\guess e}}+1} - (b+1)\guess e >  1-\frac{1}{\frac{1}{\guess e}} - (b+1)\guess e =  1-\guess e - (b+1)\guess e = 1 - (b+2)\guess e\,.$$
Otherwise, we get
\begin{align}
& \frac{1}{\kfinal+1} < \guess\nonumber\\
  \Updownarrow~~~\nonumber\\
& \kfinal > \frac{1}{\guess}-1 \,.
\end{align}
Plugging this into Ineq.~(\ref{ineq:St}), we get
\begin{align*}
  S_t & > \ln\left(\frac{1}{\guess}-1\right) -
  \ln\left(\frac{1}{\guess e} \right) - (b+1)\guess e\\
  & = \ln\left(\frac{\frac{1}{\guess}-1}{\frac{1}{\guess e}} \right)
  - (b+1)\guess e\\
  & = \ln\left(e-\guess e \right) - (b+1)\guess e\\
  & = 1 + \ln(1-\guess) -(b+1)\guess e \\
  & > 1-(b+2)\guess e, \text{ since } \guess < \frac12 \,.
\end{align*}
Thus, in either case, we get
$S_t > 1-(b+2)\guess e$.
Therefore, the algorithm can fit at most $\frac{(b+2)\guess e}{r_2
  \guess} = \frac{(b+2)e}{r_2}$ of the items of size $r_2 \guess$ into
its knapsack.
Since \ALG packs at most one item per iteration of the while-loop,
this means that
\begin{align*}
\ALG(\sigma)
  & \leq \kfinal - \FLOOR{\frac{1}{\guess e}} + 1 +
    \frac{(b+2)e}{r_2} \\
  & < \kfinal - \frac{1}{\guess e} + 2 + \frac{(b+2)e}{r_2} \\
  & < \frac{1}{\guess} - \frac{1}{\guess e} + 2 +
    \frac{(b+2)e}{r_2}, \text{ by Ineq.~(\ref{ineq:ktupper})} \\
  & = \frac{e-1}{\guess e} + 2 + \frac{(b+2)e}{r_2}\\
  & = r_2 \frac{e-1}{e}\OPT(\sigma) + 2 + \frac{(b+2)e}{r_2} \,.
\end{align*}
For any $r < 1$, this yields an upper bound on the competitive ratio
of $r \frac{e-1}{e}$, since for any given $r$, $ 2 + \frac{(b+2)e}{r}$
is a constant.  

This proves the first part of the theorem.
The second part of the theorem is just the contrapositive of the first part.
\end{proof}

Combining 
the positive result from
Theorem~\ref{thm:adbp} with the negative result from
Theorem~\ref{thm:generaluppersmallr}, we obtain that, if $r$ is
guaranteed to be smaller than $e$, no algorithm can be better than
\ADB for both $r < 1$ and $r > 1$.

\begin{restatable}{theorem}{ATsemi}
  \label{thm:ATsemi}
  \ADBML has a competitive ratio of
  $$c_{\ADBML}(r) =
    \begin{cases}
      \displaystyle r \cdot \frac{e-1}{e}, & \text{if } r \leq 1\\[2ex]
      \displaystyle \frac{e-r}{e},         & \text{if } 1 \leq r \leq e \\[1ex]
      \displaystyle 0,                     & \text{if } r \geq e
    \end{cases}
  $$
\end{restatable}
\begin{proof}
  The lower bounds follow from Theorem~\ref{thm:adbp}.

  The upper bound for $1 \leq r < e$ follows from
  Theorem~\ref{thm:generaluppersmallr}.
  Since \ADBML is $\frac{e-r}{e}$-competitive for $1 \leq r < e$, the
  upper bound for $r < 1$ also follows from
  Theorem~\ref{thm:generaluppersmallr}.

  For $r > e$, consider the input sequence consisting of
  $\frac{1}{r\guess}$ items of size $r\guess$.
  \OPT accepts all items and \ADBML accepts none.
\end{proof}

\subsection{Untrusted Predictions}
\label{sec:MLalg}

\subsubsection{Positive Result}
\label{positiveML}
When considering the case where the average item size is estimated to
be $\guess$, and the accurate value is $a=r\cdot \guess$, we
consider two cases, $r>1$ and $r<1$. In either case, we have the
problem that we do not even know which case we are in, so, when large
items arrive, we have to accept some to be competitive. The algorithm
we consider when the value of $r$ is not necessarily one achieves
similar competitive ratios in both cases.
Algorithm~\ref{algorithm-ml}, \MLB, is \bALG with a different threshold
function than was used for accurate advice (and in \ADBML).

\begin{algorithm}
\begin{algorithmic}[1]
\State Define $\TF(i)=\sqrt{\frac{\guess}{2i}}$ for $i \geq 1$
\State Run \bALG, Algorithm~\ref{algorithm-bounds}
\end{algorithmic}
\caption{\ubALG, \MLB.}
\label{algorithm-ml}
\end{algorithm}

Since we need to accept larger items than in the case of accurate
advice, we need a threshold function that decreases faster than the
threshold function used in Section~\ref{sec:advice}, in order not to risk
filling up the knapsack before the small items arrive.  Therefore, it may
seem surprising that we are using a threshold function that decreases
as $\frac{1}{\sqrt{i}}$, when the threshold function of
Section~\ref{sec:advice} decreases as $\frac1i$.  However, the
$\frac1i$-function of the algorithm for accurate advice is essentially
offset by $\frac{1}{ae}$.

We prove a number of more or less technical results before stating the
positive results for $r \leq 1$ (Theorem~\ref{thm:mlrlessthan1}) and
$r \geq 1$ (Theorem~\ref{thm:mlrgreaterthanone}).

\begin{restatable}{lemma}{sizeupper}
\label{lemma:sizeupper}
For any $k \geq 1$, the total size of the $k$ largest items accepted
by \MLB is at most $\sqrt{2k\guess}$.
\end{restatable}

\begin{proof}
  By the test in \MLB, as soon as $i$ items of size greater than
  $\TF(i+1)$ have been accepted, no more items larger than $\TF(i+1)$ are
  accepted after that.  Thus, for each $i \geq 0$, at most $i$ items
  of size greater than $\sqrt{\frac{\guess}{2(i+1)}}$ are accepted.  This
  means that the $i$th largest item accepted by \MLB has size at most
  $\sqrt{\frac{\guess}{2i}}$.  Thus, the total size of the $k$ largest
  accepted items is bounded by
$$\sum_{i=1}^{k}
\sqrt{\frac{\guess}{2i}}\leq  \sqrt{\frac{\guess}{2}}\int_0^{k} \frac{1}{\sqrt{i}}di
= \sqrt{\frac{\guess}{2}} \cdot 2\sqrt{k}=\sqrt{2k\guess},$$
since $f(i)=\frac{1}{\sqrt{i}}$ is a decreasing function.
\end{proof}

\begin{restatable}{corollary}{lowerhighlevel}
  \label{cor:lowerhighlevel}
  If \MLB rejects an item based on the level being too high, it has
  accepted at least $\lfloor\frac{1}{2\guess}\rfloor$ items.
\end{restatable}

\begin{proof}
  If \MLB has accepted $k$ items when it receives an item with a size
  no larger than the current bound, $\TF(i+1)$, that does not fit in the
  knapsack, then by Lemma~\ref{lemma:sizeupper}, $\sqrt{2(k+1)\guess} > 1$.
  Now,
$$\sqrt{2(k+1)\guess} > 1 \Leftrightarrow k > \frac{1}{2\guess}-1 \Rightarrow k
\geq \left\lfloor \frac{1}{2\guess} \right\rfloor\,.$$
\end{proof}

The following corollary implies that \MLB never rejects an item based
on the level being too high if $r>2$.  This is because $r>2$ means
that the items in \OPT are relatively large compared to $\guess$.  Since
\OPT accepts the smallest items of the sequence, it means that the
sequence contains relatively few small items.  Thus, the algorithm
reserves space for small items that never arrive.

\begin{restatable}{corollary}{levelreject}
  \label{levelreject}
  If \MLB rejects an item based on the level being too high,
   $\MLB(\sigma) >   \frac{r}{2}\OPT(\sigma)-1$.
\end{restatable}
\begin{proof}
  By Corollary~\ref{cor:lowerhighlevel},
   $$\MLB(\sigma) > \frac{1}{2\guess}-1 =
  \frac{r}{2}\cdot \frac{1}{r\guess} -1 \geq  \frac{r}{2}\OPT(\sigma)-1\,.$$
\end{proof}

\begin{restatable}{lemma}{ratiocalc}
  \label{lemma:ratiocalc}
  Assume that $r, \guess, q > 0$, $i \geq 0$, and
  $\ell < \sqrt{\frac{2(i+1)}{\guess}}$.    If
$$ 2r\sqrt{2\guess(i+1)}-2r\guess(i+1)+q-2 \leq 0,$$ then 
$$ \frac{(i+1)+\left(
    \frac{1}{r\sqrt{2\guess(i+1)}} -1\right)\ell}{\frac{\ell}{r\sqrt{2\guess(i+1)}}} > \frac{q}{2}\,.$$
\end{restatable}

\begin{proof}
  \[\begin{array}{cl}
      &  2r\sqrt{2\guess(i+1)}-2r\guess(i+1)+q-2 \leq 0 \\
      \Updownarrow \\
      &  2r\sqrt{2\guess(i+1)}-2+q \leq 2r(i+1)\guess \\
      \Updownarrow \\
      &  \sqrt{\frac{2(i+1)}{\guess}}(2r\sqrt{2\guess(i+1)}-2+q) \leq 2r(i+1)\sqrt{2\guess(i+1)} \\
      \Downarrow \\
      &  \ell(2r\sqrt{2\guess(i+1)}-2+q) <2r(i+1)\sqrt{2\guess(i+1)}, \text{ since } \ell < \sqrt{\frac{2(i+1)}{\guess}} \\
      \Updownarrow \\
      &  \ell q< 2r(i+1)\sqrt{2\guess(i+1)}+2\ell-2\ell r\sqrt{2\guess(i+1)} \\
      \Updownarrow \\
      &  \frac{\ell q}{\sqrt{2\guess(i+1)}}< 2r(i+1)+\frac{2\ell}{\sqrt{2\guess(i+1)}}-2\ell r \\
      \Updownarrow \\
      &  \frac{\ell q}{\sqrt{2\guess(i+1)}}< 2r((i+1)+\frac{\ell}{r\sqrt{2\guess(i+1)}}-\ell)  \\
      \Updownarrow \\
      &   \frac{q}{2} < \frac{(i+1)+\left(
        \frac{1}{r\sqrt{2\guess(i+1)}} -1\right)\ell}{\frac{\ell}{r\sqrt{2\guess(i+1)}}}
\end{array}\]
\end{proof}

\begin{restatable}{lemma}{slowerlupper}
  \label{lemma:slowerlupper}
  Assume that $\OPT$ accepts $\ell$ items larger than
  $\sqrt{\frac{\guess}{2(i+1)}}$ and $s$ items of size at most
  $\sqrt{\frac{\guess}{2(i+1)}}$, $i \geq 0$.  Then, the following
  inequalities hold:
    \begin{enumerate}
    \item
      $\displaystyle s >
      \ell\left(\frac{1}{r\sqrt{2\guess(i+1)}}-1\right)$ \label{slower}
    \item $\displaystyle \ell < \sqrt{\frac{2(i+1)}{\guess}}$ \label{lupper}
    \end{enumerate}
\end{restatable}
\begin{proof}
  Since \OPT's accepted items have average size $a$, we have that
  \[r\guess=a > \frac{\ell\cdot\sqrt{\frac{\guess}{2(i+1)}}+s\cdot 0}{\ell+s},\]
  and, equivalently,
\begin{equation*}
  s > \ell\left(\frac{1}{r\sqrt{2\guess(i+1)}}-1\right)\,.
  \end{equation*}
In addition, since \OPT accepts $\ell$ items larger than
$\sqrt{\frac{\guess}{2(i+1)}}$,
$\ell\sqrt{\frac{\guess}{2(i+1)}}<1$, so
\begin{equation*}
  \ell < \sqrt{\frac{2(i+1)}{\guess}}\,.
  \end{equation*}
\end{proof}

\begin{restatable}{theorem}{mlrgreaterthanone}
  \label{thm:mlrgreaterthanone}
  For all request sequences $\sigma$, such that $r\geq 1$,
  $$\MLB(\sigma)
  \geq \frac{1}{2r}\OPT(\sigma)-1.$$
\end{restatable}
\begin{proof}
  By Corollary~\ref{levelreject}, if \MLB rejects an item in $\sigma$
  due to the knapsack not having room for the item,
  $\MLB(\sigma) \geq \frac{r}{2}\OPT(\sigma)-1\geq
  \frac{1}{2r}\OPT(\sigma)-1$ for $r\geq 1$.

  Now, suppose that \MLB does not reject any item due to it not
  fitting in the knapsack. If \MLB is not optimal, it must reject due to
  the size of the item.

  Let $i_t$ denote the final value of $i$ when the algorithm is run.
  This means that \MLB has accepted $i_t$ items of size greater than
  $\sqrt{\frac{\guess}{2(i_t+1)}}$. We perform a case analysis based on
  whether this value is smaller or larger than $r\guess$.

\paragraph*{Case~1: $r\geq \frac{1}{\sqrt{2\guess(i_t+1)}}$}
In this case, $i_t+1\geq \frac{1}{2r^2\guess}$ and $\OPT(\sigma) \leq
\frac{1}{r\guess} \leq \sqrt{\frac{2(i_t+1)}{\guess}}$.
Thus,
  \[
     \frac{\MLB(\sigma)+1}{\OPT(\sigma)}
  \geq
  \frac{i_t+1}{\sqrt{\frac{2(i_t+1)}{\guess}}}  = \sqrt{\frac{\guess(i_t+1)}{2}}
  \geq \sqrt{\frac{\guess\frac{1}{2r^2\guess}}{2}}=\frac{1}{2r}\,.
\]
  Therefore,  $\MLB(\sigma)
  \geq \frac{1}{2r}\OPT(\sigma)-1$.

\paragraph*{Case~2: $r<\frac{1}{\sqrt{2\guess(i_t+1)}}$}

  Suppose \OPT accepts $\ell$ items larger than
  $\sqrt{\frac{\guess}{2(i_t+1)}}$ and $s$ items of size at most
  $\sqrt{\frac{\guess}{2(i_t+1)}}$.  Note that \MLB also accepts the $s$
  items of size at most $\sqrt{\frac{\guess}{2(i_t+1)}}$, since we are in
  the case where it does not reject items because of the knapsack
  being too full.
 
  Given the input sequence $\sigma$, we consider the ratio
\[
     \frac{\MLB(\sigma)+1}{\OPT(\sigma)}
  \geq
  \frac{(i_t+1)+s}{\ell+s}\,.
\]

The result follows if this ratio is always at least $\frac{1}{2r}$.

\paragraph*{Subcase~2a: $i_t+1 \geq \frac{1}{2\guess}$}

In this case, $\MLB(\sigma) \geq i_t \geq \frac{1}{2\guess}-1$, while
$\OPT(\sigma) \leq \frac{1}{r\guess}$. Thus,
$\MLB(\sigma)\geq \frac{r}{2}\OPT(\sigma)-1$.

\paragraph*{Subcase~2b: $i_t+1 < \frac{1}{2\guess}$}

By Ineq.~\ref{slower} of Lemma~\ref{lemma:slowerlupper}, and since
$\frac{\MLB(\sigma)+1}{\OPT(\sigma)}\leq 1$,
\[
     \frac{\MLB(\sigma)+1}{\OPT(\sigma)}
  \geq
  \frac{(i_t+1)+s}{\ell+s}
  \geq
  \frac{(i_t+1)+\left(
    \frac{1}{r\sqrt{2\guess(i_t+1)}} -1\right)\ell}{\frac{\ell}{r\sqrt{2\guess(i_t+1)}}}\,.
\]

We will show that this is at least $\frac{1}{2r}$.

From our case conditions, $i_t+1<\frac{1}{2\guess}$ and
$1\leq r<\frac{1}{\sqrt{2\guess(i_t+1)}}$, we get that
$\frac{1}{r^2} >2\guess(i_t+1)$ and $0<2\guess(i_t+1)<1$.  Consider the
function $$f(r)=2r\sqrt{2\guess(i_t+1)}-2r\guess(i_t+1)+\frac{1}{r}-2\,.$$
Taking the derivative with respect to $r$ gives
\begin{align*}
  f'(r) & =2\sqrt{2\guess(i_t+1)}-2\guess(i_t+1)-\frac{1}{r^2}.
\end{align*}
Setting this equal to zero and solving for~$r$, we find
\[r^{\ast} = \frac{1}{\sqrt{2\sqrt{2\guess(i_t+1)}-2\guess(i_t+1)}}.\]
The possible maximum value for $f(r)$ in the range for $r$ is
then at $1$, $r^{\ast}$, or $\frac{1}{\sqrt{2\guess(i_t+1)}}$.
For all three values, $f(r)\leq 0$. The hardest (but still
simple) case is for~$r=r^{\ast}$, where
\[f(r^{\ast}) = \frac{2\sqrt{v} - v}{\sqrt{2\sqrt{v}-v}} + \sqrt{2\sqrt{v}-v} - 2,\]
where we let $v$ denote $2\guess(i_t+1)$.
Note that due to the subcase we are in, $0<v<1$.
Now,
\begin{align*}
    & \frac{2\sqrt{v} - v}{\sqrt{2\sqrt{v}-v}} + \sqrt{2\sqrt{v}-v} - 2 \leq 0 \\
      \Updownarrow~~~ \\
    & 4\sqrt{v} - 2v - 2 \leq 0 \\
    \Updownarrow~~~ \\
    & 2\sqrt{v} \leq v + 1 \\
    \Updownarrow~~~ \\
    & 4v \leq v^2 + 2v + 1 \\
    \Updownarrow~~~ \\
    & 0 \leq (v-1)^2
\end{align*}
which clearly holds.

By Ineq.~\ref{lupper} of Lemma~\ref{lemma:slowerlupper}, the result now
follows from Lemma~\ref{lemma:ratiocalc} with $q=\frac1r$.
\end{proof}

\begin{restatable}{theorem}{mlrlessthanone}
    \label{thm:mlrlessthan1}
  For all request sequences $\sigma$, such that $r< 1$,
  $$\MLB(\sigma) \geq \frac{r}{2}\OPT(\sigma)-1.$$
\end{restatable}
\begin{proof}
  The proof follows that of the previous theorem.

\noindent {\bf Case 1.} $i_t+1 \geq \frac{1}{2\guess}$.
  
  Since $\ell\geq i_t+1$ (otherwise \MLB is optimal), \MLB has
  accepted at least $\frac{1}{2\guess}-1$ items, while \OPT can accept at
  most $\frac{1}{r\guess}$. Thus,
  $\MLB(\sigma)\geq \frac{r}{2}\cdot\frac{1}{r\guess}-1\geq
  \frac{r}{2}\OPT(\sigma)-1$.

\noindent {\bf Case 2.} $i_t+1<\frac{1}{2\guess}$.

By Ineq.~\ref{slower} of Lemma~\ref{lemma:slowerlupper} and since
$\frac{\MLB(\sigma)+1}{\OPT(\sigma)}\leq 1$,
\[
     \frac{\MLB(\sigma)+1}{\OPT(\sigma)}
  \geq
  \frac{(i_t+1)+s}{\ell+s}
  \geq
  \frac{(i_t+1)+\left(
    \frac{1}{r\sqrt{2\guess(i_t+1)}} -1\right)\ell}{\frac{\ell}{r\sqrt{2\guess(i_t+1)}}}\,.
\]

We will show that this is at least $\frac{r}{2}$.  Consider the
function $$f(r)=2r\sqrt{2\guess(i_t+1)}-2r\guess(i_t+1)+r-2\,.$$ Taking the
derivative with respect to $r$ gives
$$f'(r)=2\sqrt{2\guess(i_t+1)}-2\guess(i_t+1)+1\,,$$ which is positive, since by the case condition,
$0<2\guess(i_t+1)<1$. Thus, $f(r)$ is an increasing function for the values
of $\guess$, $i_t+1$, and $r$ considered in this case, so the maximum value
is at the maximum value of $r$, $r=1$, giving that
\[
  f(r) \leq 2r\sqrt{2\guess(i_t+1)}-2r\guess(i_t+1)+r-2<0.
\]
By Ineq.~\ref{lupper} of Lemma~\ref{lemma:slowerlupper}, the result now
follows from Lemma~\ref{lemma:ratiocalc} with $q=r$.
\end{proof}

\subsubsection{Negative Result}

In Section~\ref{sec:advice}, we showed that, even with accurate
advice, no algorithm can be better than $\frac{e-1}{e}$-competitive.
In this section, we give a trade-off in the competitive ratio attained
for different values of~$r$.

\begin{restatable}{theorem}{mllower}
  \label{mllower}
  Let $0< z\leq 2$.  No algorithm can have a competitive ratio better
  than $\frac{1}{zr}$ for every $r$ between $\frac2z$ and
  $\frac{1}{\sqrt{z\guess}}$.  Moreover, any algorithm which is
  $\frac{1}{zr}$-competitive for all $r$ in this interval has a
  competitive ratio of at most $\frac{zr}{4}$, for any positive
  $r < \frac{2}{z}$.
\end{restatable}
\begin{proof}
  We consider the adversary that gives the input sequence $\sigma_z$
  defined by Algorithm~\ref{trade-off-adversarial}.
    \begin{algorithm}
    \label{tradeoffalg}
    \begin{algorithmic}[1]
    \LeftComment{Assume $\frac{1}{q\guess}\in\NAT$}
    \State $p$ \ASSIGN $\lfloor \frac{z}{4\guess}\rfloor$
    \State $k$ \ASSIGN $0$ \linespread{1.2}\selectfont
    \While{$k\leq p-1$}
       \State $k$++
       \State Give $\left\lfloor
       \sqrt{\frac{zk}{\guess}}\right\rfloor$ items of size 
       $\sqrt{\frac{\guess}{zk}}$\label{itemsize}
       \If{\ALG has accepted fewer than $k-b$ items} \textbf{terminate} \label{terminate}
       \EndIf
    \EndWhile
    \State Give $\frac{1}{q\guess}$ items of size $q\guess$ \label{smallitems}
    \end{algorithmic}
    \caption{Adversarial sequence establishing trade-off on robustness
      versus consistency. The adversarial algorithm takes
      parameters, $z$, $q$, and $b$, such that $0 < z \leq 2$,
      $0 < q < \frac{1}{\sqrt{z\guess}}$, and $b\geq 0$.}
\label{trade-off-adversarial}
\end{algorithm}
  We begin with the second part of the theorem.
  Consider an online algorithm, \ALG, and assume that there exists a
  constant, $b$, such that
  $\ALG(\sigma) \geq \frac{1}{zr} \OPT(\sigma)-b$, for any sequence
  $\sigma$ and any $r$ such that $\frac2z \leq r \leq
  \frac{1}{\sqrt{z\guess}}$.
  Now,
  consider the adversary that gives the input sequence $\sigma_z$
  defined by Algorithm~\ref{trade-off-adversarial}.

If the adversarial algorithm terminates in Line~\ref{terminate}, then,
\ALG has accepted at most $k-b-1$ items.  In this case,
$a=\sqrt{\frac{\guess}{zk}}$, and \OPT accepts exactly the
$\left\lfloor \sqrt{\frac{zk}{\guess}}\right\rfloor$ items from the last
iteration of the while-loop. Since $a=r\guess$, $r=\sqrt {\frac{1}{zk\guess}}$, which lies between
$\sqrt{\frac{1}{zp\guess}} \geq \sqrt{\frac{1}{z\guess} \cdot \frac{4\guess}{z}} =
\frac{2}{z}$ and $\frac{1}{\sqrt{z\guess}}$.  Thus,
\begin{align*}
\ALG(\sigma_z)
& \; \leq \; k-b-1 
  \; \leq \; \frac{k-1}{\left\lfloor\sqrt{\frac{zk}{\guess}}\right\rfloor}\OPT(\sigma_z)-b 
  \; < \; \frac{k-1}{\sqrt{\frac{zk}{\guess}}-1}\OPT(\sigma_z)-b \\
& \; < \; \frac{k}{\sqrt{\frac{zk}{\guess}}}\OPT(\sigma_z)-b 
  \; = \; \sqrt{\frac{k\guess}{z}}\OPT(\sigma_z)-b 
  \; = \; \frac{1}{zr}\OPT(\sigma_z)-b \,,
\end{align*}
where the second strict inequality holds because $1$  is
added to the numerator and denominator of a positive fraction less than $1$.
This contradicts the assumption that for each $r$
  between $\frac2z$ and $\frac{1}{\sqrt{z\guess}}$, 
  $\ALG(\sigma) \geq \frac{1}{zr} \OPT(\sigma)-b$, for any sequence
  $\sigma$, when the adversarial algorithm terminates in Line~\ref{terminate}.
  Thus, the adversarial algorithm does not terminate there.
  
  If the adversarial algorithm does not terminate in
  Line~\ref{terminate}, $r=q$ and
  $\OPT(\sigma_z) = \frac{1}{q\guess} =
  \frac{1}{r\guess}$.  Moreover, for \ALG, the $i$th accepted
  item must have size at least $\sqrt{\frac{\guess}{z(i+b)}}$, for
  $1 \leq i \leq p-b$. Thus, these first $p-b$ items fill the knapsack to
  at least
$$\sum_{i=b+1}^{p}
\sqrt{\frac{\guess}{zi}}\geq \sqrt{\frac{\guess}{z}}\int_{b+1}^{p+1}
\frac{1}{\sqrt{i}}di = \sqrt{\frac{\guess}{z}} (2\sqrt{p+1}-2\sqrt{b+1}),$$
where we use that $\frac{1}{\sqrt{i}}$ is a decreasing function.

Since the items of size $r\guess$ are the smallest items of the sequence,
  this means that
\begin{align*}
  \ALG(\sigma_z)
  & \leq p+\frac{1-\sqrt{\frac{\guess}{z}} (2\sqrt{p+1}-2\sqrt{b+1})}{r\guess} \\
  & \leq \frac{z}{4\guess}+\frac{1-\sqrt{\frac{\guess}{z}} \left(2\sqrt{\frac{z}{4\guess}}-2\sqrt{b+1}\right)}{r\guess} \\
  & = \frac{z}{4\guess}+\frac{1- 1 + 2\sqrt{\frac{\guess(b+1)}{z}}}{r\guess} \\
  & =\frac{1}{r\guess}\left(\frac{zr}{4}+2\sqrt{\frac{\guess(b+1)}{z}}\right) \\
  & = \left(\frac{zr}{4}+2\sqrt{\frac{\guess(b+1)}{z}}\right)\OPT(\sigma_z)\,.
\end{align*}
As a function of $\guess$, the lower bound is
$\frac{zr}{4}+2\sqrt{\frac{\guess(b+1)}{z}}$, but the second term becomes
insignificant as $\guess$ approaches zero.

To show that the algorithm cannot
have a competitive ratio better than $\frac{1}{zr}$, for every $r$ between
$\frac{2}{z}$ and $\frac{1}{\sqrt{z\guess}}$, we consider
Algorithm~\ref{trade-off-adversarial} with the item sizes on
Line~\ref{itemsize} equal to $\sqrt{\frac{\guess}{zk}}+\eps$. Following the
proof above,
in the case where the adversarial algorithm terminates in Line~\ref{terminate},
$\ALG$'s competitive ratio is at most $\frac{1}{zr}$, for $r$ in this
range. However, for small enough $\guess$, the adversarial algorithm must at
some point terminate in
Line~\ref{terminate}, since otherwise the knapsack would be over-filled:
Similar to the calculations above, 
the first $p-b$ items fill the knapsack to at least
\begin{align*}
\sum_{i=b+1}^{p} \sqrt{\frac{\guess}{zi}}+\eps
& \geq \sqrt{\frac{\guess}{z}}\int_{b+1}^{p+1} \frac{1}{\sqrt{i}}di +(p-b)\eps \\
& = \sqrt{\frac{\guess}{z}} \left(2\sqrt{p+1}-2\sqrt{b+1}\right) +(p-b)\eps \\
& \geq  \sqrt{\frac{\guess}{z}} \left(2\sqrt{\frac{z}{4\guess}-1+1}-2\sqrt{b+1}\right) +\left(\frac{z}{4\guess}-1-b\right))\eps \\
& = 1-2\sqrt{b+1} +\left(\frac{z}{4\guess}-1-b\right)\eps\,.
\end{align*}
Since $\guess$ can be arbitrarily small, the last term can dominate
the second term, giving a result larger than $1$. The result of
terminating in Line~\ref{terminate} is the same as above,
$\ALG(\sigma_z)< \frac{1}{zr}\OPT(\sigma_z)-b$, giving a
contradiction.
\end{proof}

Setting $z=2$ in Theorem~\ref{mllower} demonstrates a Pareto-like
trade-off between consistency and robustness for \MLB:

\begin{corollary}
  No algorithm can have a competitive ratio better than $\frac{1}{2r}$
  for every $r$ between $1$ and $\frac{1}{\sqrt{2\guess}}$.
  Moreover, any algorithm which is $\frac{1}{2r}$-competitive for all
  $r$ between $1$ and $\frac{1}{\sqrt{2 \guess}}$ has a competitive
  ratio of at most $\frac{r}{2}$ for any positive $r<1$.
\end{corollary}

\section{Advice Complexity}
\label{numberofbits}
In this section, we briefly consider the Online Unit Cost Knapsack Problem
in terms of advice complexity, concentrating on upper bounds,
following the techniques in~\cite{BKKR14} and many other articles on
advice complexity including~\cite{ADKRR18,CFL15}.
One assumes that a certain number of bits are
available to approximate actual values (that might not be small integers).

The advice given in the algorithms \ADBML and \MLB is
a prediction for the value, $a$, representing the average size of an
item that \OPT accepts, and it could have some error. One could use
\ADBML in the advice complexity setting, assuming that an oracle gives
two values: $z$, the number of zeros after the binary point in the
binary representation of $a$, followed by $s$, the next $k$ bits of
$a$. In this case, the prediction $\guess$ given for $a$ is
$\frac{s}{2^{z+k}}$.  (The numerator should be thought of as the
value, e.g., if $s$ is the bits $1101$, the value is~$13$.)  Since the
high order bit of $s$ is $1$, this value is at least
$\frac{1}{2^{z+1}}$.  The error in the prediction, $\guess$, is only
due to the missing low order bits (assumed, possibly incorrectly, to
be zero).  The missing bits represent a number less
than~$\frac{1}{2^{z+k}}$.  Thus, the ratio, $r$, in $a=r\guess$ is in
the range~$1\leq r\leq 1+\frac{1}{2^{k-1}}$.

By Theorem~\ref{thm:adbp},
we can use the algorithm \ADBML (with the modification that it calculates
$\guess=\frac{s}{2^{z+k}}$ after reading and decoding the advice) and obtain
that for all $\sigma$,
$$\ADBML(\sigma) \geq \frac{e-1-\frac{1}{2^{k-1}}}{e}\OPT(\sigma)\,.$$
Note that the length of the advice is independent of the length of the
request sequence, though dependent on the values in that sequence.
The value, $z$, and the bitstring, $s$, must be specified using
self-delimiting encoding, since we do not know how many bits are used
for them. For example, $\lceil \log (z+1)\rceil$ could be written in
unary ($\lceil \log (z+1)\rceil$ ones, followed by a zero) before
writing $z$ itself in binary. Treating $s$ similarly, at most
$2(k+\lceil \log (z+1)\rceil +1)$ bits are used.

Since \OPT can be viewed as accepting a prefix of the sequence of
items sorted in non-decreasing order of size, there is another obvious
type of advice to give.  Let the advice be $k$-bit approximations to
both the size of the largest item that \OPT accepts, $s$, and the
fraction, $t$, of the knapsack not filled with items of size strictly
smaller than~$s$. The approximation to $s$ can be given using the
technique above, specifying the number of leading zeros first and then
$k$ significant bits, For $t$, we do not use the count of leading
zeros and simply use the $k$ most significant bits.
There are two reasons that it is necessary to give
the fraction of the knapsack not filled with items of at most this
size. One reason is that,
even if the exact value of $s$ was given, it is unknown if \OPT accepts one
or many items of that size, and these ``large'' items could come before any
smaller ones. The other reason is that, since the size of this largest
accepted item is rounded
down, there may be many items that \OPT accepts that are larger than
this (though never any item as large as~$s+\frac{1}{2^{z+k}}$). Thus,
it can be necessary to accept many items larger than $s$, and we need
to know how much space we can use for this, or if space should be saved for
many items much smaller than~$s$.
The algorithm will accept all items that are smaller than $s$, which is the optimal behavior on those items (so in
the worst case for the performance ratio, no such items arrive).
Thus, we are only interested in items of size between $\frac{s}{2^{z+k}}$ and
$\frac{s+1}{2^{z+k}}$ and can calculate a bound on the competitive ratio just from the algorithm's
and \OPT's performance on items in that range. Since the algorithm does not accept all items in the
worst case, we may assume that there are enough items in this size range that it 
rejects some. Under this assumption,
the algorithm accepts at least
$\FLOOR{\frac{\frac{t}{2^k}}{\frac{s+1}{2^{z+k}}}}$ and \OPT accepts at most
$\FLOOR{\frac{\frac{t+1}{2^k}}{\frac{s}{2^{z+k}}}}$. For an asymptotic result, ignoring the rounding down on the algorithm's performance, this gives a performance
ratio of at least $$\frac{t\cdot s}{(s+1)(t+1)}\geq \frac{2^k\cdot 2^{z+k}}
{(2^k+1)(2^{z+k}+1)}\geq \frac{2^{2k}}{(2^k +1)^2}.$$
Since we approximate two values, we need twice as much advice as for the
first approach, that is
$4(k+\lceil \log (z+1)\rceil +1)$ bits of advice.
The competitive ratio with this approach is better than that of the first
approach, but it also uses more advice.

With respect to optimality, we note that the lower bound of $\log n$
from~\cite{BKKR14} for the general Knapsack Problem cannot be used
directly here, since the items used in their sequences all have
size~$1$, so the weights are very important.  In contrast to the upper
bounds proven above, we prove that for optimality, the number of
advice bits needed is a function of $n$, at least~$\log (n/3)$.
Consider the set of input sequences defined to have length $n$ as
follows: Let $n=3k$ and $0\leq \ell\leq k$. Then $I_\ell$ consists of
(in the order listed)
\begin{itemize} 
\item  $k$ items of size $\frac{1}{k}$,
\item $2(k-\ell)$ items of size $\frac{1}{2k}$,
  \item $2\ell$ items of size $1$.
\end{itemize}
Suppose for the sake of contradiction that \ALG is optimal on all of
these sequences and never reads $\log(n/3)$ bits of advice.  \OPT
accepts $\ell$ items of size $\frac{1}{k}$ and then $2(k-\ell)$ items
of size $\frac{1}{2k}$, completely filling up the knapsack with
$2k-\ell$ items.  Intuitively, the advice needs to say how many of the
first $k$ items to accept. Since there are $n/3$ sequences in all and
fewer than $\log(n/3)$ bits of advice, there are at least two of the
sequences $I_j$ and $I_{j'}$ for which \ALG receives the same
advice. Thus, \ALG accepts the same number, say $j^*$, items of size
$\frac{1}{k}$ on both $I_j$ and $I_{j'}$. Without loss of generality,
assume that $j^*\not= j'$. If $j^* > j'$, then \ALG can accept only
$2(k-j^*)$ items of size $\frac{1}{2k}$. In all,
$\ALG(I_{j'}) \leq j^*+2(k-j^*) <2k-j' = \OPT(I_{j'})$. If $j^* < j'$,
then $\ALG(I_{j'}) \leq j^* +2(k-j') < 2k-j' = \OPT(I_{j'})$. Thus,
\ALG is not optimal on $I_{j'}$, giving a contradiction.



\bibliography{refs}

\begin{thebibliography}{10}

\bibitem{AEMP22}
Sara Ahmadian, Hossein Esfandiari, Vahab Mirrokni, and Binghui Peng.
\newblock Robust load balancing with machine learned advice.
\newblock In {\em 2022 {ACM-SIAM} Symposium on Discrete Algorithms (SODA)},
  pages 20--34. {SIAM}, 2022.

\bibitem{A21}
Spyros Angelopoulos.
\newblock Online search with a hint.
\newblock In {\em 12th Innovations in Theoretical Computer Science Conference
  (ITCS)}, volume 185 of {\em LIPIcs}, pages 51:1--51:16. Schloss Dagstuhl -
  Leibniz-Zentrum f{\"{u}}r Informatik, 2021.

\bibitem{ADJKR20}
Spyros Angelopoulos, Christoph D{\"{u}}rr, Shendan Jin, Shahin Kamali, and
  Marc~P. Renault.
\newblock Online computation with untrusted advice.
\newblock In {\em 11th Innovations in Theoretical Computer Science Conference
  (ITCS)}, volume 151 of {\em LIPIcs}, pages 52:1--52:15. Schloss Dagstuhl -
  Leibniz-Zentrum f{\"{u}}r Informatik, 2020.

\bibitem{ADKRR18}
Spyros Angelopoulos, Christoph D{\"{u}}rr, Shahin Kamali, Marc~P. Renault, and
  Adi Ros{\'{e}}n.
\newblock Online bin packing with advice of small size.
\newblock {\em Theory of Computing Systems}, 62(8):2006--2034, 2018.

\bibitem{AK21}
Spyros Angelopoulos and Shahin Kamali.
\newblock Contract scheduling with predictions.
\newblock In {\em 35th {AAAI} Conference on Artificial Intelligence (AAAI),
  33rd Conference on Innovative Applications of Artificial Intelligence (IAAI),
  11th Symposium on Educational Advances in Artificial Intelligence (EAAI)},
  pages 11726--11733. {AAAI} Press, 2021.

\bibitem{AKS21}
Spyros Angelopoulos, Shahin Kamali, and Kimia Shadkami.
\newblock Online bin packing with predictions.
\newblock {\em ArXiv}, 2021.
\newblock arXiv:2102.03311 [cs.DS].

\bibitem{AKZ21}
Spyros Angelopoulos, Shahin Kamali, and Dehou Zhang.
\newblock Online search with best-price and query-based predictions.
\newblock {\em ArXiv}, 2021.
\newblock arXiv:2112.01592 [cs.DS].

\bibitem{ACEPS20}
Antonios Antoniadis, Christian Coester, Marek Eli{\'{a}}s, Adam Polak, and
  Bertrand Simon.
\newblock Online metric algorithms with untrusted predictions.
\newblock In {\em 37th International Conference on Machine Learning (ICML)},
  volume 119 of {\em Proceedings of Machine Learning Research}, pages 345--355.
  {PMLR}, 2020.

\bibitem{AGKK20}
Antonios Antoniadis, Themis Gouleakis, Pieter Kleer, and Pavel Kolev.
\newblock Secretary and online matching problems with machine learned advice.
\newblock In {\em 33rd Annual Conference on Neural Information Processing
  Systems (NeurIPS)}, pages 7933--7944. Curran Associates, Inc., 2020.

\bibitem{BMRS20}
Etienne Bamas, Andreas Maggiori, Lars Rohwedder, and Ola Svensson.
\newblock Learning augmented energy minimization via speed scaling.
\newblock In {\em 33rd Annual conference on Neural Information Processing
  Systems (NeurIPS)}, pages 15350--15359. Curran Associates, Inc., 2020.

\bibitem{BMS20}
Etienne Bamas, Andreas Maggiori, and Ola Svensson.
\newblock The primal-dual method for learning augmented algorithms.
\newblock In {\em 33rd Annual conference on Neural Information Processing
  Systems (NeurIPS)}, pages 20083--20094. Curran Associates, Inc., 2020.

\bibitem{BGGJ22}
Siddhartha Banerjee, Vasilis Gkatzelis, Artur Gorokh, and Billy Jin.
\newblock Online nash social welfare maximization with predictions.
\newblock In {\em 2022 {ACM-SIAM} Symposium on Discrete Algorithms (SODA)},
  pages 1--19. {SIAM}, 2022.

\bibitem{BCKPV22}
Nikhil Bansal, Christian Coester, Ravi Kumar, Manish Purohit, and Erik Vee.
\newblock Learning-augmented weighted paging.
\newblock In {\em 2022 {ACM-SIAM} Symposium on Discrete Algorithms (SODA)},
  pages 67--89. {SIAM}, 2022.

\bibitem{BCKP20}
Aditya Bhaskara, Ashok Cutkosky, Ravi Kumar, and Manish Purohit.
\newblock Online learning with imperfect hints.
\newblock In {\em 37th International Conference on Machine Learning (ICML)},
  volume 119 of {\em Proceedings of Machine Learning Research}, pages 822--831.
  {PMLR}, 2020.

\bibitem{BKKR14}
Hans-Joachim B\"ockenhauer, Dennis Komm, Richard Kr\'alovi\v{c}, and Peter
  Rossmanith.
\newblock The online knapsack problem: Advice and randomization.
\newblock {\em Theoretical Computer Science}, 527:61--72, 2014.

\bibitem{BFKLM17}
Joan Boyar, Lene~M. Favrholdt, Christian Kudahl, Kim~S. Larsen, and Jesper~W.
  Mikkelsen.
\newblock {Online Algorithms with Advice: A Survey}.
\newblock {\em ACM Computing Surveys}, 50(2):1--34, 2017.
\newblock Article No.~19.

\bibitem{BFLN01}
Joan Boyar, Lene~M. Favrholdt, Kim~S. Larsen, and Morten~N.\ Nielsen.
\newblock The competitive ratio for on-line dual bin packing with restricted
  input sequences.
\newblock {\em Nordic Journal of Computing}, 8:463--472, 2001.

\bibitem{CFL15}
Marie~G. Christ, Lene~M. Favrholdt, and Kim~S. Larsen.
\newblock {Online Multi-Coloring with Advice}.
\newblock {\em Theoretical Computer Science}, 596:79--91, 2015.

\bibitem{CJS16}
Marek Cygan, {\L{}}ukasz Je{\.{z}}, and Jir{\'{\i}} Sgall.
\newblock Online knapsack revisited.
\newblock {\em Theory of Computing Systems}, 58, 2016.

\bibitem{GP19}
Sreenivas Gollapudi and Debmalya Panigrahi.
\newblock Online algorithms for rent-or-buy with expert advice.
\newblock In {\em 36th International Conference on Machine Learning (ICML)},
  volume~97 of {\em Proceedings of Machine Learning Research}, pages
  2319--2327. {PMLR}, 2019.

\bibitem{IKQP21}
Sungjin Im, Ravi Kumar, Mahshid~Montazer Qaem, and Manish Purohit.
\newblock Non-clairvoyant scheduling with predictions.
\newblock In {\em 33rd {ACM} Symposium on Parallelism in Algorithms and
  Architectures (SPAA)}, pages 285--294. {ACM}, 2021.

\bibitem{IKQP21knapsack}
Sungjin Im, Ravi Kumar, Mahshid~Montazer Qaem, and Manish Purohit.
\newblock Online knapsack with frequency predictions.
\newblock In {\em Pre-Proceedings of the 34th Annual Conference on Neural
  Information Processing Systems (NeurIPS)}, 2021.

\bibitem{IMMR20}
Piotr Indyk, Frederik Mallmann-Trenn, Slobodan Mitrović, and Ronitt Rubinfeld.
\newblock Online page migration with {ML} advice.
\newblock {\em ArXiv}, 2020.
\newblock arXiv:2006.05028 [cs.DS].

\bibitem{JPS20}
Zhihao Jiang, Debmalya Panigrahi, and Kevin Sun.
\newblock Online algorithms for weighted paging with predictions.
\newblock In {\em 47th International Colloquium on Automata, Languages, and
  Programming (ICALP)}, volume 168 of {\em LIPIcs}, pages 69:1--69:18. Schloss
  Dagstuhl - Leibniz-Zentrum f{\"{u}}r Informatik, 2020.

\bibitem{KPP04}
Hans Kellerer, Ulrich Pferschy, and David Pisinger.
\newblock {\em Knapsack problems}.
\newblock Springer, 2004.

\bibitem{K19}
Rohan Kodialam.
\newblock Optimal algorithms for ski rental with soft machine-learned
  predictions.
\newblock {\em ArXiv}, 2019.
\newblock arXiv:1903.00092 [cs.DS].

\bibitem{KA18}
Arvind Kumar and Bashir Alam.
\newblock Task scheduling in real time systems with energy harvesting and
  energy minimization.
\newblock {\em Journal of Computational Science}, 14(8):1126--1133, 2018.

\bibitem{LLMV20}
Silvio Lattanzi, Thomas Lavastida, Benjamin Moseley, and Sergei Vassilvitskii.
\newblock Online scheduling via learned weights.
\newblock In {\em 31st {ACM-SIAM} Symposium on Discrete Algorithms (SODA)},
  pages 1859--1877. {SIAM}, 2020.

\bibitem{LMRX21}
Thomas Lavastida, Benjamin Moseley, R.~Ravi, and Chenyang Xu.
\newblock {Learnable and Instance-Robust Predictions for Online Matching, Flows
  and Load Balancing}.
\newblock In {\em 29th Annual European Symposium on Algorithms (ESA)}, volume
  204 of {\em LIPIcs}, pages 59:1--59:17. Schloss Dagstuhl -- Leibniz-Zentrum
  f{\"u}r Informatik, 2021.

\bibitem{LHL19}
Russell Lee, Mohammad~H. Hajiesmaili, and Jian Li.
\newblock Learning-assisted competitive algorithms for peak-aware energy
  scheduling.
\newblock {\em ArXiv}, 2020.
\newblock arXiv:1911.07972 [cs.DS].

\bibitem{LMHLSL21}
Russell Lee, Jessica Maghakian, Mohammad~H. Hajiesmaili, Jian Li, Ramesh~K.
  Sitaraman, and Zhenhua Liu.
\newblock Online peak-aware energy scheduling with untrusted advice.
\newblock In {\em 12th {ACM} International Conference on Future Energy Systems
  (e-Energy)}, pages 107--123. {ACM}, 2021.

\bibitem{LX21}
Shi Li and Jiayi Xian.
\newblock Online unrelated machine load balancing with predictions revisited.
\newblock In {\em 38th International Conference on Machine Learning (ICML)},
  volume 139 of {\em Proceedings of Machine Learning Research}, pages
  6523--6532. {PMLR}, 2021.

\bibitem{LV18}
Thodoris Lykouris and Sergei Vassilvitskii.
\newblock Competitive caching with machine learned advice.
\newblock In {\em 35th International Conference on Machine Learning (ICML)},
  volume~80, pages 3302--3311. {PMLR}, 2018.

\bibitem{LV21}
Thodoris Lykouris and Sergei Vassilvitskii.
\newblock Competitive caching with machine learned advice.
\newblock {\em Journal of the {ACM}}, 68(4):24:1--24:25, 2021.

\bibitem{M-SV95}
Alberto Marchetti{-}Spaccamela and Carlo Vercellis.
\newblock Stochastic on-line knapsack problems.
\newblock {\em Mathematical Programming}, 68:73--104, 1995.

\bibitem{MV17}
Andres~Mu{\~{n}}oz Medina and Sergei Vassilvitskii.
\newblock Revenue optimization with approximate bid predictions.
\newblock In {\em 30th Annual Conference on Neural Information Processing
  Systems (NIPS)}, pages 1858--1866. Curran Associates, Inc., 2017.

\bibitem{M20}
Michael Mitzenmacher.
\newblock {Scheduling with Predictions and the Price of Misprediction}.
\newblock In {\em 11th Innovations in Theoretical Computer Science Conference
  (ITCS)}, volume 151 of {\em LIPIcs}, pages 14:1--14:18. Schloss
  Dagstuhl--Leibniz-Zentrum fuer Informatik, 2020.

\bibitem{M21}
Michael Mitzenmacher.
\newblock Queues with small advice.
\newblock In {\em SIAM Conference on Applied and Computational Discrete
  Algorithms (ACDA)}, pages 1--12, 2021.

\bibitem{MV20}
Michael Mitzenmacher and Sergei Vassilvitskii.
\newblock Algorithms with predictions.
\newblock {\em ArXiv}, 2020.
\newblock arXiv:2006.09123 [cs.DS].

\bibitem{PSK18}
Manish Purohit, Zoya Svitkina, and Ravi Kumar.
\newblock Improving online algorithms via {ML} predictions.
\newblock In {\em 31st Annual Conference on Neural Information Processing
  Systems (NeurIPS)}, pages 9661--9670. Curran Associates, Inc., 2018.

\bibitem{R20}
Dhruv Rohatgi.
\newblock Near-optimal bounds for online caching with machine learned advice.
\newblock In {\em 31st {ACM-SIAM} Symposium on Discrete Algorithms (SODA)},
  pages 1834--1845. {SIAM}, 2020.

\bibitem{RM21}
Daan Rutten and Debankur Mukherjee.
\newblock A new approach to capacity scaling augmented with unreliable machine
  learning predictions.
\newblock {\em ArXiv}, 2021.
\newblock arXiv:2101.12160 [cs.DS].

\bibitem{WL20}
Shufan Wang and Jian Li.
\newblock Online algorithms for multi-shop ski rental with machine learned
  predictions.
\newblock In {\em 19th International Conference on Autonomous Agents and
  Multiagent Systems (AAMAS)}, pages 2035--2037. International Foundation for
  Autonomous Agents and Multiagent Systems, 2020.

\bibitem{W20}
Alexander Wei.
\newblock Better and simpler learning-augmented online caching.
\newblock In {\em Approximation, Randomization, and Combinatorial Optimization.
  Algorithms and Techniques (APPROX/RANDOM)}, volume 176 of {\em LIPIcs}, pages
  60:1--60:17. Schloss Dagstuhl - Leibniz-Zentrum f{\"{u}}r Informatik, 2020.

\bibitem{WZ20}
Alexander Wei and Fred Zhang.
\newblock Optimal robustness-consistency trade-offs for learning-augmented
  online algorithms.
\newblock {\em ArXiv}, 2020.
\newblock arXiv:2010.11443 [cs.DS].

\bibitem{ZSHW20}
Ali Zeynali, Bo~Sun~Mohammad Hajiesmaili, and Adam Wierman.
\newblock Data-driven competitive algorithms for online knapsack and set cover.
\newblock In {\em 35th AAAI Conference on Artificial Intelligence (AAAI)},
  2021.

\end{thebibliography}

\end{document}